\documentclass[12pt,draftcls,onecolumn]{IEEEtran}
\usepackage{graphicx,color,epsfig,rotating,dblfloatfix}
\usepackage{amssymb,amsmath,amsthm}
\usepackage{subfigure, dsfont}
\usepackage{epstopdf}
\usepackage{verbatim}
\usepackage{cite}

\newtheorem{lthm}{Theorem}

\newtheorem{lem}[lthm]{Lemma}
\newtheorem{proposition}{Proposition}

\usepackage{algorithm}
\usepackage{algpseudocode}
\usepackage{tensor}
\usepackage{pifont}
\usepackage{flushend}
\usepackage{cuted}
\makeatletter
\def\BState{\State\hskip-\ALG@thistlm}
\makeatother

\begin{document}

\title{\LARGE{Percentile Policies for Tracking of Markovian Random Processes with Asymmetric Cost and Observation}}
\author{Parisa Mansourifard$^*$,  ~\IEEEmembership{Student Member}~IEEE, Tara Javidi$^\dagger$, ~\IEEEmembership{Senior Member}~IEEE, Bhaskar Krishnamachari$^*$, ~\IEEEmembership{Senior Member}~IEEE
\thanks{$^*$Parisa Mansourifard and Bhaskar Krishnamachari  are with Ming Hsieh Department of Electrical Engineering, University of Southern California,
Los Angeles 90089 CA USA (email: { parisama@usc.edu, bkrishna@usc.edu}).  }
\thanks{$^\dagger$Tara Javidi is with the Department of Electrical and Computer Engineering, University of California, San Diego, La Jolla, CA 92093 USA.(email: {tjavidi@ucsd.edu})}
        }

\maketitle

\begin{abstract}
Motivated by wide-ranging applications such as video delivery over networks using Multiple Description Codes (MDP), congestion control, rate adaptation, spectrum sharing, provisioning of renewable energy, inventory management and retail, we study the state-tracking of a Markovian random process with a known transition matrix and a finite ordered state set. The decision-maker must select a state as an action at each time step in order to minimize the total expected (discounted) cost. The decision-maker is faced with asymmetries both in cost and observation:
in case the selected state is less than the actual state of the Markovian process, an under-utilization cost occurs and only partial observation about the actual state (\textit{i.e.} an implicit lower bound characteristic on the actual state) is revealed; otherwise, the decision incurs an over-utilization cost and reveals full information about the actual state.
We can formulate this problem as a Partially Observable Markov Decision Process (POMDP) which can be expressed as a dynamic program (DP) based on the last full observed state and the time of full observation. This formulation determines the sequence of actions to be taken
between any two consecutive full observations of the actual state, in order to minimize the total expected (discounted) cost.
However, this DP grows exponentially in the number of states, with little hope for a computationally feasible solution.
We present an interesting class of computationally tractable policies with a percentile threshold structure. A generalization of binary search, this class of policies attempt at any given time to reduce the uncertainty by a given percentage.
Among all percentile policies, we search for the one with the minimum expected cost.
 The result of this search is a heuristic policy which we evaluate through numerical simulations. 
 We show that it outperforms the myopic policies and under some conditions performs close to the optimal policies. 
  Furthermore, we derive a lower bound on the cost of the optimal policy which can be computed with low complexity and give a measure for how close our heuristic policy is to the optimal policy.
\end{abstract}

%\begin{IEEEkeywords}
%Tracking a Markovian process, POMDP, Threshold policy, Percentile policy
%\end{IEEEkeywords}

\section{\textbf{Introduction}} \label{sec:intro}

In networks with uncertain demands (or resources) the available resources (or demands) must be allocated carefully, to address a trade-off between inefficiencies arising from over-utilization and under-utilization, and to gather information useful for future decisions in case of time-dependent variations.
For instance, in communication networks, one of the important goals is satisfying the application traffic demands using available resources such as bandwidth, energy, storage, etc. There could be uncertainty in either demand or resource. 
%Consequently, there are two possible problems that are valuable to consider: one is resource allocation under demand uncertainty, and the other is demand allocation under resource uncertainty. 
Under demand uncertainty allocating more resources than what is demanded can result in an over-utilization cost while the under-utilization immediately incurs a cost. 
Similarly, under resource uncertainty, excess demand results in congestion (over-utilization) while shortage of demand leaves valuable resources under-utilized.

One prominent example is video delivery over an unreliable network where Multiple Description Coding (MDC) is used, specially for real-time applications in which retransmission of lost packets are not practical \cite{goyal2001multiple,wang2002error,wang2005multiple,pereira2003multiple,akyol2007flexible}.
MDC combats the uncertainty over network resources via encoding source information using multiple independently decodable complementary bitstreams called descriptions. 
If all of the descriptions are received at the receiver, they can be decoded combined, with the highest level of quality and when a non-empty subset of them is received, the information could still be decoded providing an acceptable level of quality.
This technique allows for a graceful degradation in the quality of the image.
To achieve this, MDC introduces some redundancy in each description which will be costly (in terms of transmission resources) whenever all of the descriptions are received \cite{goyal2001multiple}.
%Note that higher number of successfully received descriptions results in higher quality for the client, thus higher reward for the decision-maker, but on the other hand higher number of descriptions results in higher redundancy (reduced compression efficiency \cite{wang2005multiple}) and more complex coding/decoding processes. 
In this example, the uncertain available resource is the number of descriptions that network can handle to deliver to the receiver and the decision-maker must decide on the actual number of descriptions to code the video stream. % (which corresponds to the demand). %The over-utilization cost corresponds to the number of descriptions not received in the receiver and the under-utilization cost corresponds to the lower quality of the decoded video at the receiver. 
%A fundamental design problem in MDC is for the source to choose the number of descriptions to code the video stream for a given network period. If the packet loss process of the network can be modeled as a Markovian random process, the actual number of descriptions that can go through the network successfully which we refer to as the state of the network forms a time-homogeneous finite-state Markov Chain. The full observation about the state of the network is revealed only if the number of descriptions were higher than what the network can handle.
%The aim of this work is to estimate how many descriptions could go through the network or equivalently to track the state of the network and select the number of descriptions as close as possible to the state in order to maximize/minimize the total expected reward/cost. 
%Depending on selecting the number of descriptions higher or lower than the actual state, different costs have to be paid and different observations about the actual state are revealed. 
A fundamental design problem in MDC is for the source to choose the number of descriptions to code the video stream for a given network period. If the packet loss process of the network can be modeled as a Markovian random process, the actual number of descriptions that can go through the network successfully, which we refer to as the state of the network, forms a time-homogeneous finite-state Markov Chain. The full observation about the state of the network is revealed only if the number of descriptions were higher than what the network can handle. 
The aim of this work is to estimate how many descriptions could go through the network or equivalently to track the state of the network and select the number of descriptions as close as possible to the state in order to maximize/minimize the total expected reward/cost.

%The tracking of Markovian random process considered in this paper can also be used to solve the problem of congestion control in the networks and inventory control problems with perishable goods. In general, the problem of tracking a Markov chain subject to asymmetric cost and information is known to be a Partially Observable Markov Decision Problem (POMDP) whose solution in principle can be characterized via an appropriately constructed dynamic program. However, such characterization, due to its computational complexity, provides little practical insight. More precisely, even in case of finite underlying state and finite actions, a POMDP in general is known to be P-SPACE hard \cite{papadimitriou1987complexity}. 

The tracking of Markovian random process considered in this paper can be used to solve the above problem of selecting the optimal number of MDC descriptions. Other applications to which the tracking problem considered can be applied, as we discuss in related works, are congestion control in the networks and inventory control problems with perishable inventories. In general, the problem of tracking a Markov chain subject to asymmetric cost and information is known to be a Partially Observable Markov Decision Problem (POMDP). As such, the solution in principle can be characterized via an appropriately constructed dynamic program. However, such characterization, due to its computational complexity, provides little practical insight. More precisely, even in case of finite underlying state and finite actions, a POMDP in general is known to be P-SPACE hard \cite{papadimitriou1987complexity}.

To summarize, our main contributions to the tracking of Markovian random processes are as follows:
\begin{itemize}
\item Representing the class of optimal policies in terms of a sequence of actions in between two consecutive full observation of the state. 
\item Introducing a new class of percentile policies with percentile threshold structure. This class of policies can be used to upper bound the expected cost of the optimal policy.
\item Providing a lower bound on the expected cost of the optimal policy. Similar to the upper bound, the lower bound can be constructed recursively.  
\item Presenting the best computable percentile policy, called Finite Resolution Percentile (FRP), that in simulations outperforms the myopic policy. Furthermore, we utilize the numerical simulations of the optimal policy for sufficiently small horizons to investigate the performance of the proposed heuristic policies. In particular, consider a sufficiently small horizon, we show numerically that this policy performs close to the optimal policy. 
\item Numerically assessing the optimality gap associated with the best computable percentile policy using the upper and lower bounds on the expected cost of the optimal policy. 
\end{itemize}

The remainder of the paper is organized as follows: The related works are presented in Section \ref{sec:work}. The problem formulation is given in Section \ref{sec:problem}. We present the optimal policy via a sequence-based dynamic programming and its complexity in Section \ref{sec:seqDP}.
We introduce the percentile policy and present two heuristic percentile policies in Section \ref{sec:percentile}.
An upper bound and a lower bound on the expected cost of the optimal policy are given in Section \ref{sec:perf}.
Section \ref{sec:belief} shows the extension of the policies to the case where the initial belief about the actual state is given instead of the full observation.
The numerical results are presented in Section \ref{sec:simulation}. Section \ref{sec:conclusion} provides conclusions and an outlook for future works and finally the equivalent belief-based formulation and all of the proofs are given in Appendices.

\section{\textbf{Related Works}}\label{sec:work}

Here, we present three different sets of literature related to our work. The first set of literature corresponds to MDC, \textit{e.g.} \cite{goyal2001multiple,wang2002error,wang2005multiple,pereira2003multiple,akyol2007flexible} and the references therein.
Some of the works in the literature, considers only two-description coding 
\cite{goyal2001multiple,wang2002error,wang2005multiple}, but in some cases such as Peer-to-Peer video streaming more than two-descriptions are reasonable, see \textit{e.g.} \cite{akyol2007flexible}.
We refer the readers to \cite{kazemi2014review} as a review of deffierent MDC techniques since in this work, we do not focus on the MDC techniques. We focus on the modeling of the cost functions and how to decide about the number of descriptions based on partial or full observations about the state of the network, \textit{i.e.} the number of available paths with no loss, which is a function of the packet loss process of the network and evolves as a Markovian process.

The second set of literature related to our work %is the resource (demand) allocation in networks with uncertain demand (resource). In networks with uncertain demands (or resources) the available resources (or demands) must be allocated carefully, to address a trade-off between inefficiencies arising from over-utilization and under-utilization, and to gather information useful for future decisions in case of time-dependent variations. For instance, in communication networks, one of the important goals is satisfying the application traffic demands using available resources such as bandwidth, energy, storage, etc. There could be uncertainty in either demand or resource. Consequently, there are two possible problems that are valuable to consider: one is resource allocation under demand uncertainty, and the other is demand allocation under resource uncertainty. If the demand is uncertain, by allocating more resources than what is demanded, in general, an over-utilization cost is incurred. On the other hand, the under-utilization immediately incurs a cost especially if some demands are left unsatisfied. Likewise, in case of resource uncertainty, if the allocated demands exceed the available resources, it may result in an increased delay or congestion in the network.
%One example 
is congestion/rate control in which a transmitter must select the transmission rate at the transport or link layer to utilize the available bandwidth, which varies randomly due to the dynamic nature of traffic load imposed by other users on the network \cite{mansourifard2013bayesian,mansourifard2015, johnston2006opportunistic,laourine2010betting}. This scenario is of particular interest when designing wireless networks where resources are subject to random variations and uncertainty. 
The structure of the optimal policies has been established for the simpler special cases of optimizing transmissions over a Gilbert-Elliott (two-state) channel in \cite{johnston2006opportunistic,laourine2010betting, wu2012online}. 
In this paper, we study a generalization of these works to more than two states.

Beside the networking application, this work could have a broad range of applications, \textit{e.g.}, incentive spectrum-sharing \cite{berry2013nature,leng2014microeconomic}, the provisioning of renewable energy resources \cite{nair2012energy}, retailer-consumer interaction \cite{huang2015incentive}, etc. For instance, in the spectrum-sharing application \cite{leng2014microeconomic} in cellular wireless networks, the goal is to allocate/share the resource (Base-station capacities) between two cellular providers in a single-cell in order to satisfy uncertain bandwidth demands inside the cell. 
As another example, in the electricity markets with renewable energy resources  \cite{nair2012energy}, the uncertainty is on both the demand and the renewable energy source side and the goal is to find the optimal procurement strategy.
In the retailer-consumer interaction, the retailer needs to track consumers' privacy sensitivity to offer personalized advertisements (coupons) \cite{huang2015incentive}. The authors in \cite{huang2015incentive} assume that the privacy state of the consumer evolves as a Markovian process and the retailer must choose as an action the advertisement privacy levels in order to minimize the total expected cost. They consider the two-state two-action and multi-state two-action cases and model the problem as a POMDP. Our work provides a direction to generalize their work to the multi-state multi-action case.

 Last but not least, the third set of problems related to our work is the perishable inventory control problem in operations research management literature as it applies to the problem of optimal inventory control to meet uncertain demands for a perishable product \cite{bensoussan2007multiperiod,qin2011newsvendor,negoescu2008optimal, besbes2010implications,bensoussan2005optimal,bensoussan2008inventory, lu2005inventory, chen2010bounds}. 
In these problems, the demand for some good is assumed to follow a stochastic process and at the beginning of each decision epoch the decision-maker decides on the inventory level (\textit{i.e.} how many items to store) in order to satisfy the demand. Mapping the Markov demand to a hidden state and the inventory level to the selected action, the inventory control problem with perishable good is equivalent to that of tracking with asymmetric cost and information. Most of the works in the inventory control literature, \textit{e.g.} \cite{ding2002censored,bensoussan2009technical}, assume that the demand process is independent and identical distributed (i.i.d) at different time steps, with the exception of \cite{bensoussan2007multiperiod}. With this simplifying assumption, the optimal policy is easily shown to coincide with a myopic policy which minimizes the immediate expected cost. 
The most closely related work to ours is \cite{bensoussan2007multiperiod} in which an inventory management problem with memory (Markovian) demand process is considered. %However in \cite{bensoussan2007multiperiod}, the authors consider the setting where the resources and actions are both continuous, as well as the case where the resources are discrete but the actions remain continuous. For these settings, where the decision-maker can make very high resolution tracking decisions, 
In their work, some structural properties of the optimal actions relative to the myopically optimal actions are obtained. 
In our work, in contrast, we focus on the design and analysis of a class of heuristic policies. In particular, we introduce the class of \textbf{percentile policies} and 
evaluate their performance. In addition, we present a lower bound on the cost of the optimal policy which can be computed with low complexity and give a measure for how close our heuristic policies are to the optimal policy.

\section{\textbf{Problem Formulation}}\label{sec:problem}
We consider a discrete-time finite-state Markovian process whose state, denoted by $B_t$ at time step $t$, selected from the finite state set $\mathcal{M}=\{0,1,,...,M\}$, and evolves based on a known transition matrix, over a finite horizon, $T$. 
The transition probabilities of the actual states $B_t$ over time are assumed to be known and stationary and indicated by an $(M+1)\times (M+1)$ transition probability matrix, $P$. The elements of the matrix are $P_{i,j}=Pr(B_{t+1}=j|B_t=i), i,j\in \mathcal{M}, \forall t$ which indicates the probability of moving from the state $i$ at a time step to the state $j$ at the next time step.

The objective is to select the sequential actions (policy) among the states in $\mathcal{M}$ such that the total expected (discounted) cost accumulated over the finite horizon $T$ is minimized.
At each time step $t$, the decision-maker selects a state as an action, denoted by $r_t$, based on the history of observations and pays a cost which is a function of the selected state and the actual state $B_t$.
In particular, selecting an action higher than the actual state incurs an over-utilization cost; while in contrast, selecting an action less than the actual state causes a distinctly under-utilization cost.
The immediate cost paid at time step $t$ is a piece-wise linear function of the difference between the selected state and the actual state, given by:
\begin{align}\label{eq:rwd}
C(B_t;r_t)=\begin{cases}
 c_u(r_t-B_t)  &\text{if}\ r_t>B_t \\
 c_l(B_t-r_t)  &\text{if}\ r_t \leq B_t, \\
\end{cases}
\end{align}
where $c_u$ and $c_l$ are the over-utilization and under-utilization cost coefficients, respectively.\footnote{For certain problems, \textit{e.g.} \cite{mansourifard2013bayesian,wu2012online}, the cost may be more naturally expected as $f(B_t)+C(B_t,r_t)$ where $f(.)$ is a function of only $B_t$ and is not under the control, or the problem may be defined as a reward maximization. However these are mathematically equivalent.} 
Note that the immediate cost is not observable when $r_t< B_t$.
Additionally, in our models the actions may result in information asymmetry about the hidden state of the system: selecting an action that is greater than the state has the side benefit of revealing the exact state of the actual (hidden) state; in contrast, selecting an action below the actual (hidden) state does not provide direct information about the state. Note that the absence of full state observation in itself results in an implicit lower bound characterization on the state. 

The goal is to select actions $r_1,...,r_T$ in order to minimize the total expected (discounted) cost over the horizon $T$, given by:
\begin{align}\label{eq:minprob1}
%\min_{\pi} J_T^{\pi}(s_0)=\min_{\pi} 
\mathbb{E}\{\sum_{t=1}^T {\beta ^{t-1}  C(B_t;r_t)}|s_0\},
\end{align}
where $0\leq \beta \leq 1$ denotes the discount factor and $s_0$ is the initial full observed state. 

Figure \ref{fig:samplepath} shows an example where a sample path of $B_t$ and a sequence of actions are selected based on a given policy. %As shown in this figure, when the action exceeds the actual state, the full observation about the actual state is revealed. 
The under-utilization and over-utilization costs for each time step is given on the figure.

\begin{figure}[t]
    \centerline{\includegraphics[width=.8\linewidth]{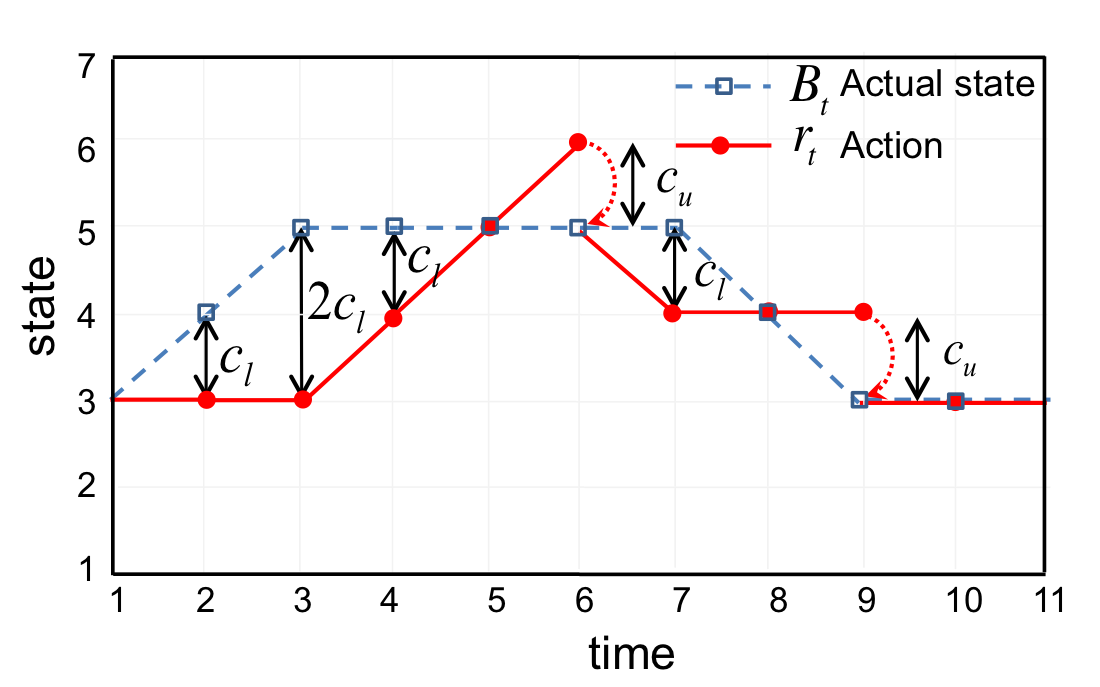}}
    \caption{An example tracking of a sample path of the actual states $B_t$, using the actions $r_t$. In this example, at time steps 6 and 9 the actions exceed the actual state and thus the full observation is revealed. At other time steps, only the partial observation is available. The under-utilization and over-utilization costs are shown on the figure.
    For instance, the costs are $C(B_3=5,r_3=3)=2c_l$, $C(B_6=5,r_6=6)=c_h$ and $C(B_{10}=3,r_{10}=3)=0$.}
    \label{fig:samplepath}
\end{figure}

Since we do not get full observation all the time, we formulate our problem within a POMDP-based framework. We present the POMDP formulation based on the pair of $(s,t)$ where $s$ and $t$ are the state and the time where that very last time the state was fully revealed (full observation). The decision-maker must select the whole sequence of actions to be taken, for the given observed state of $s$, from time $t+1$ until the next full observation.
Note that after time $t$, the whole sequence of actions with the length of $T-t$ %for $\tau=t+1,t+2,...,T$ 
must be determined.

The observation at each time step is a function of the actual state and the selected action as follows: If $r_t>B_t$ a full observation about $B_t$ is revealed and full observed state is updated to $B_t$ and the time of full observation is updated to $t$. At the next time step, the policy changes the sequence correspondingly.
Otherwise, if $r_t \leq B_t$ only partial observation about the actual state is revealed that it is higher than or equal to $r_t$ and the policy continues with the action sequence corresponding to the last full observation.
For instance, at Fig. \ref{fig:samplepath}, initially, $s=3$ and $t=1$ and the actions are taken correspondingly. After exceeding the actual state at time step 6 the full observation is updated to $s=5$ and $t=6$. Thus the corresponding actions are taken at time steps $7\leq t\leq 9$ and after time step 9, the full observation state and time are updated to $s=3$ and $t=9$ and their corresponding actions are selected at $10\leq t$. 

\section{\textbf{Optimal Policy via Dynamic Programming}}\label{sec:seqDP}

Let us consider the class $\Pi$ of policies which map the last full observed state $s$ and the time of observation $t$ to an action sequence,
\textit{i.e.}, 
\begin{align}
\Pi: \mathcal{M}\times [0,T-1] \rightarrow \mathcal{M}^{T-t}, \nonumber
\end{align}
where $\mathcal{M}^{T-t}$ is the space of possible sequences of length $T-t$. It it straight forward to show that without loss of optimality the search for an optimal policy can be restricted to $\Pi$.

Note that, alternatively, any policy $\pi \in \Pi$ can be represented by an $(M+1)\times T$ structure whose $(s,t)$ element is a sequence of length $T-t$ where elements themselves are selected from $\mathcal{M}$. 
Next we use this representation to compute the performance of the policy recursively.

Before writing the recursion, we consider $_{\underline{a}}P(s_1,s_2,\tau)$ probability of going from state $s_1 \in \mathcal{M}$ to $s_2 \in \mathcal{M}$ in $\tau$ steps never crossing below a given sequence $\underline{a}=(a_1,...,a_{\tau})$, defined in (\ref{eq:Ps1s2}). 
We also define $_{\underline{a}}\Gamma(s_1,s_2,\tau)$ as the expected discounted cost-to-go to be paid after the full observation of $s$ at time step $t$ up to the next full observation, as given by (\ref{eq:Gammas1s2})\footnote{This is a generalization of the notion of the taboo probability \cite{cox1977theory}.}.

\begin{figure*}
\begin{align}\label{eq:Ps1s2}
_{\underline{a}}P(s_1,s_2,\tau)=
\begin{cases}
 P_{s_1,s_2},   &\text{if}\ \tau=1,\\
 \sum_{j_1=a_1}^M \sum_{j_2=a_2}^M ... \sum_{j_{\tau-1}=a_{\tau}}^M P_{s_1,j_1} \times \prod_{n=2}^{\tau-1} P_{j_{n-1},j_n} \times P_{j_{\tau-1},s_2},   &\text{if}\ \tau>1. 
 \end{cases}
\end{align}
\end{figure*}

\begin{figure*}
\begin{align}\label{eq:Gammas1s2}
_{\underline{a}}\Gamma&(s_1,t)=\sum_{\tau=1}^{T-t} \beta^{\tau-1} [c_u \sum_{i=0}^{a_{\tau}-1}\  _{\underline{a}}P(s,i,\tau)(a_{\tau}-i)
+c_l \sum_{i=a_{\tau}}^M \ _{\underline{a}}P(s,i,\tau)(i-a_{\tau})].
\end{align}
\end{figure*}

Note that under policy $\pi$ and given the full observation of the state $s$ at time $t$, sequence of actions $\underline{\pi}_{s,t}=(\pi_{s,t}(1),...,\pi_{s,t}(T-t))$ are used until the next full observation of the state. In other words, one can compute the expected cost under policy $\pi$ as follows:
\begin{align}\label{eq:Wpi}
W^{\pi}_{T-1}(s)&=\  _{\underline{\pi}_{s,T-1}}\Gamma(s,T-1),\nonumber\\
W^{\pi}_{t}(s)&=\ _{\underline{\pi}_{s,t}}\Gamma(s,t)\nonumber\\
&+\sum_{\tau=1}^T \beta^{\tau} \sum_{s'=0}^{\underline{\pi}_{s,t}(\tau)-1}\ _{\underline{\pi}_{s,t}}P(s,s',\tau) W^{\pi}_{t+\tau}(s'),\nonumber\\
&\forall t=0,...,T-2.
\end{align}
Note that the total expected cost given in (\ref{eq:minprob1}) is equivalent to $W^{\pi}_{0}(s_0)$. 
Since the horizon $T$ and the state space $\mathcal{M}$ are both finite, the policy space $\Pi$ is finite, it is possible to define an optimal policy $\pi^{opt}$ to minimize the expected cost:
\begin{align*}%\label{eq:optmin}
\pi^{opt}=\arg \min_{\pi} W^{\pi}_{0}(s_0).
\end{align*}
which provides the optimal sequences of $\underline{\pi}^{opt}_{s,t}$ with the length of $T-t$ for any $s\in \mathcal{M}$ and $t=0,...,T-1$, such that
\begin{align*}
W^{\pi^{opt}}_{t}(s)=\arg \min_{\pi\in \Pi} W^{\pi}_{t}(s).
\end{align*}

\subsection{\textbf{Optimal Policy: An Example and Complexity}}

As an example, we compute the optimal policy for the small horizon of $T=5$ numerically, for the parameters of $M=2$, $T=7$, $c_u=c_l=1$, $\beta=1$, and the following transition matrix:
\begin{align}\footnotesize
P=\begin{bmatrix}
    .8 & .2 & 0  \\
    .1 & .6 & .3 \\
     0 & .4 & .6 \\
\end{bmatrix}.\label{eq:Pran}
\end{align}
For any full observation of $s\in \{0,1,2\}$ at time $t\in \{0,...,6\}$, the optimal policy selects the corresponding action sequence given in Table \ref{fig:opttable}.
Unfortunately for large horizons, even numerically, the optimal action sequences are intractable.

\begin{table}[htb]
    \caption{The action sequences $\underline{\pi}^{opt}_{s,t}$ selected by the optimal policy for each full observed state $s$ at time $t$, for $M=2$, $T=7$, $c_u=c_l=1$, $\beta=1$, and $P$ given in (\ref{eq:Pran}).}
    \centerline{\includegraphics[width=.8\linewidth]{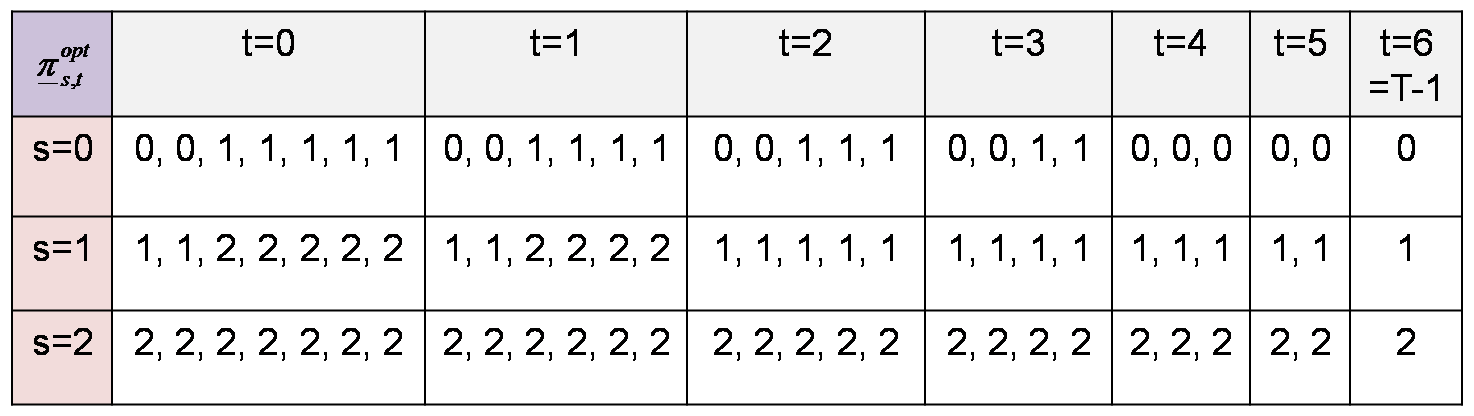}}
    \label{fig:opttable}
\end{table}

The recursive equations are not solvable for large horizons since the number of possible sequences is exponentially large. The complexity of the optimal policy, except for some trivial cases discussed in Appendix \ref{sec:special}, is given in the following proposition.

\begin{proposition}\label{prop:compDP}
The complexity of computing the optimal policy is equal to $\Theta((M+1)^{T+1})$. {\footnote{The notation $f(n)=\Theta(g(n))$ means that $f$ is bounded both above and below by $g$ asymptotically}}
\end{proposition}

\proof
The computation needed to find the optimal sequence for each observed state at observation time of $t$ is equal to $(M+1)^{T-t}$, since we have $(M+1)$ states, the total complexity can be obtained as $(M+1) \sum_{t=0}^{T-1} (M+1)^{T-t}$ which equals to $\Theta((M+1)^{T+1})$.
\endproof
 
 As shown in the above proposition, the computation for the optimal policy grows exponentially with the time horizon and thus, we could only compute the optimal policy for very small horizons.

\section{\textbf{Percentile Policy}}\label{sec:percentile}

Since finding the optimal policy given in previous section is computationally intractable for large horizons, we introduce a new class of policies, \textit{i.e.} percentile policies. This class of policies is inspired by the presentation of the optimal policy based on the sequences corresponding to the full observed states and the observation times.

We define any percentile policy based on a table of thresholds $h_{s,t}$ for $s\in \mathcal{M}$, $t=0,...,T-1$ which could be any real number between 0 and 1. Based on the given thresholds, the action sequences $\underline{\pi}^{per}_{s,t}$ corresponding to any observed state $s$ and the observation time $t$ are given by the following recursive equations:
\begin{align}\label{eq:rptha}
\underline{\pi}^{per}_{s,t}(\tau)&=\min \{r \in \mathcal{M}: \frac{\sum_{i=0}^{r} [_{\underline{\pi}^{per}_{s,t}}P(s,i,\tau)]}{\sum_{i=0}^{M} [_{\underline{\pi}^{per}_{s,t}}P(s,i,\tau)]} \geq h_{s,t}\},\nonumber\\
&\ \forall \tau=1,...,T-t.
\end{align}
In other words, the action at each time corresponds to the lowest state above a given percentile, $h_{s,t}$, of the propagated belief at that time. The algorithm to compute the action sequences and the expected costs of a percentile policy for given thresholds $h_{s,t}$ is presented in Algorithm \ref{alg:FRP}.
Note that since the thresholds are given for any $s$ and $t$, we do not need to search for the best threshold and thus $\mathcal{H}_{s,t}$ given in Algorithm \ref{alg:FRP} are simply equal to $\{h_{s,t}\}$.
Note that the corresponding expected (discounted) costs of a percentile policy (denoted by $W^{per}_t(s)$) for the given thresholds $h_{s,t}$ can be computed from (\ref{eq:Wpi}) using the action sequences of $\underline{\pi}^{per}_{s,t}$ generated in (\ref{eq:rptha}).

\begin{algorithm}
\scriptsize
\caption{Percentile Policy}\label{alg:FRP}
\begin{algorithmic}[1]
\BState{$//$ Initializations}
\State Parameters $c_u$, $c_l$, $\beta$, $P$, $T$, $M$
\State  \State $s\in \{0,...,M\}$ and $t=0,...,T-1$, and the threshold sets $\mathcal{H}_{s,t}\subset [0,1]$
\State Define $\bar C(b,r):=c_u \sum_{i=0}^{r-1}(r-i) b(i)+c_l \sum_{i=r}^{M}(i-r) b(i)$
\BState{$//$ Main loops}
\State $t \gets T-1$
\For{each state $s \in \{0,...,M\}$ }
\State vector $b \gets P_{s,.}$, $s$-th row of matrix $P$
\State let $w_{min}=\infty$
\For{$h_{s,t} \in \mathcal{H}_{s,t}$}
\State let $r = \min\{r: \sum_{i=0}^r b(i) \geq h_{s,t}\}$
\State $\underline{\pi}_{s,t}^{per}(1) \gets r$
%\State let $\bar C=c_u \sum_{i=0}^{r-1}(r-i) b(i)+c_l \sum_{i=r}^{M}(i-r) b(i)$
\State $\Gamma  \gets \bar C(b,r)$
\If {$\Gamma \leq w_{min}$}
\State $w_{min} \gets \Gamma$
\State $h^{per}_{s,t}  \gets h_{s,t}$
\EndIf
\EndFor
\State $W^{per}_t(s) \gets w_{min}$
\EndFor
\While{$t> 0$}
\State $t \gets t-1$
\For{each state $s \in \{0,...,M\}$ }
\State let $w_{min}=\infty$
\For{$h_{s,t} \in \mathcal{H}_{s,t}$}
\State $b \gets P_{s,.}$, $s$-th row of matrix $P$
\State let $r = \min\{r: \sum_{i=0}^r b(i) \geq h_{s,t}\}$
\State $\underline{\pi}_{s,t}^{per}(1) \gets r$
\State $\Gamma  \gets \bar C(b,r)$
\State let $\Lambda = 0$, and $m=1$
\For{$t'=t+1:1:T-1$}
\State $\Lambda \gets \Lambda+  \beta^{t'-t} m \sum_{i=0}^{r-1} b(i) W^{per}_{t'}(i)$
\State $m \gets m * \sum_{i=r}^{M} b(i)$
\State $b \gets \frac{\sum_{i=r}^{M} b(i)P_{i,.}}{\sum_{i=r}^{M} b(i)}$
\State let $r = \min\{r: \sum_{i=0}^r b(i) \geq h_{s,t}\}$
\State $\underline{\pi}_{s,t}^{per}(t'-t) \gets r$
%\State {let $\bar C=c_u \sum_{i=0}^{r-1}(r-i) b(i)$\\
%$+c_l \sum_{i=r}^{M}(i-r) b(i)$}
\State $\Gamma \gets \Gamma +\beta^{t'-t} m \bar C(b,r)$
\EndFor
\State $w \gets \Gamma+\Lambda$
\If {$w \leq w_{min}$}
\State $w_{min} \gets w$
\State $h^{per}_{s,t}  \gets h_{s,t}$
\EndIf
\EndFor
\State $W^{per}_t(s) \gets w_{min}$
\EndFor
\EndWhile
\end{algorithmic}
\end{algorithm}

The complexity of the percentile policy is obtained in the following proposition (see Appendix \ref{app:comp} for the proof).
\begin{proposition}\label{prop:compPTf}
The complexity of computing the action sequences of the percentile policy with given thresholds is equal to:
\begin{align}\nonumber
\Theta(T^2(M+1)^{3}).
\end{align}
\end{proposition}

Next, we present two specific percentile policies and evaluate their performances in Section \ref{sec:simulation} numerically.
The first heuristic has a fixed threshold at all time $t$ after observing any state $s$; while the second optimizes to some extend the thresholds across time and space.

\subsection{\textbf{Myopic Policy}}
One simple heuristic policy is the myopic policy which minimizes the immediate expected cost ignoring its impact on the future cost. The myopic policy, for a given belief $b_{t}$ at any time $t$, selects an action which minimizes the immediate expected cost as follows: 
\begin{align}\label{eq:myopic}
\pi^{myopic}(b_{t})
:=\arg\min_{r \in \mathcal{M}} \bar C(b_{t};r).
\end{align}
where $\bar C(b_{t};r)$ is the immediate expected cost obtained by taking expectation of (\ref{eq:rwd}), as follows:
\begin{align}
\bar C(b_{t};r)&=\sum_{i\in \mathcal{M}} {b_{t}(i) C(i;r)}  = c_l \sum_{i=r}^M {b_{t}(i)(i-r)}\nonumber\\
&+c_u \sum_{i=0}^{r-1} {b_{t}(i)(r-i)}. \label{eq:vfnR}
\end{align}
Therefore, the myopic policy can be derived as given in the following proposition (see Appendix \ref{app_lproofs} for proof).

\begin{proposition}\label{C_mybb1}
For the given belief $b_{t}$, the myopic policy is given by
\begin{align} \label{eq:rmc1}
\pi^{myopic}(b_{t})=\min \{r \in \mathcal{M}: \sum_{i=0}^{r} {b_{t}(i)} \geq h^m\},
\end{align}
where $h^m=\frac{c_l}{c_l+c_u}$.
\end{proposition}

Note that the belief propagation and the action sequences of the myopic policy can be obtained from (\ref{eq:rptha}) using the thresholds $h_{s,t}=h^m$ for any $s\in \mathcal{M}$ and $t=0,...,T-1$.

\subsection{\textbf{Finite Resolution Percentile Policy}}\label{sec:FRP}

Next, we introduce a percentile policy which chooses the thresholds providing the minimum cost-to-go for the given observed states and observation times among all possible thresholds in the resolution set of $\mathcal{H}\subset [0,1]$. We call this policy Finite Resolution Percentile (FRP) policy, which chooses the thresholds as follows:
\begin{align}
h^{FRP}_{s,t}&= \arg \min_{h_{s,t} \in \mathcal{H}} W^{per}_t(s), \forall s\in\mathcal{M}, \ t=0,...,T-1,\nonumber
\end{align}
where the action sequence $\underline{\pi}^{per}_{s,t}$ is a function of $h_{s,t}$ given in (\ref{eq:rptha}). 
Let $\underline{\pi}^{FRP}_{s,t}$ denote the action sequences generated from (\ref{eq:rptha}) using the best thresholds $h^{FRP}_{s,t}$.

Note that FRP is the best heuristic among all percentile policies if one limits its attention to the choice of the thresholds from a finite set of values between 0 and 1.
In numerical simulations, we use trial and error method with a resolution of $\Delta$ among all possible values $\mathcal{H}=\{0,\Delta,2\Delta,...,1\}\cup \{h^m\}$ to find the best thresholds.
The algorithm to compute the action sequences and the expected costs of FRP policy is given in Algorithm \ref{alg:FRP} with the same threshold sets of $\mathcal{H}_{s,t}=\mathcal{H}$.

The complexity of FRP is given in the following proposition. 

\begin{proposition}\label{prop:compPT}
The complexity of computing the actions of FRP policy is equal to:
\begin{align}
\Theta(T^2(M+1)^{3}|\mathcal{H}|),
\end{align}
where $|\mathcal{H}|$ indicates the size of the threshold set $\mathcal{H}$ which for $\mathcal{H}= \{0,\Delta,2\Delta,...,1\}\cup \{h^m\}$ equals to $1/\Delta+2$. 
\end{proposition}

The proof of the above Proposition is trivial from the Proposition \ref{prop:compPTf} since we need to compute the expected cost-to-go corresponding to each threshold in $\mathcal{H}$ and select the threshold which achieves the minimum expected cost, for each observation state and time. Thus, the complexity is $|\mathcal{H}|$ times the complexity of Proposition \ref{prop:compPTf}.
Comparing the above proposition and Proposition \ref{prop:compDP}, we conclude that computing FRP policy is polynomial and much faster than solving DP of optimal policy with the exponential complexity. And as we shall see, it can outperform the myopic policy.

\section{\textbf{Upper Bound and Lower Bound on Cost of Optimal Policy}}\label{sec:perf}

As the definition of the optimal policy suggested, any percentile policy can provide an upper bound on the total expected (discounted) cost of the optimal policy: %and among them FRP policy provides the tightest upper bound.
\begin{align}\nonumber
W^{\pi^{opt}}_0(s_0) \leq W^{per}_0(s_0).
\end{align}
In addition, we present a lower bound on the expected cost of the optimal policy which equals to the expected cost of a genie-aided decision-maker whose state observation is available (with one unit of time delay).
The advantage of obtaining a (sufficiently tight) lower bound is that it allows us to use this lower bound to evaluate the performance of FRP and myopic policies.

\begin{proposition}\label{prop:FO}
The total expected (discounted) cost of the optimal policy for our POMDP problem is lower bounded by the total expected (discounted) cost of a genie policy which gets full observation about the actual states at all time steps with one unit of time delay (we call this genie policy FO), under the same initially observed state $s_0$, \textit{i.e.},
\begin{align}\label{eq:FO}
W^{\pi^{opt}}_0(s_0) \geq W^{FO}_0(s_0),
%J^{opt}_T(s_0) \geq J^{FO}_T(s_0)
\end{align}
where
%\begin{align}%\label{eq:FOPercentile}
%&J^{FO}_T(s_0)=W^{FO}(s_0,0),
%\end{align}
$W^{FO}(s,t)$ for $s\in \mathcal{M}$ and $t=0,...,T-1$ can be computed recursively as follows:
\begin{align}
W^{FO}_{T-1}(s)&=\bar C(P_{s,.},\pi^{myopic}(P_{s,.})), \nonumber\\
W^{FO}_t(s)&=\bar C(P_{s,.}, \pi^{myopic}(P_{s,.}))+\beta \sum_{i=0}^{M} P_{s,i} W^{FO}_{t+1}(i),\label{eq:WFO}
\end{align}
where $\pi^{myopic}(P_{s,.})$ denotes the action selected by the myopic policy presented in Proposition \ref{C_mybb1}, and the immediate expected cost $\bar C(P_{s,.},\pi^{myopic}(P_{s,.})$ can be obtained by (\ref{eq:vfnR}).
\end{proposition}

We use the equivalent belief-based formulation presented in Appendix \ref{app_beliefDP} to prove the above proposition in Appendix \ref{app:PBFO}.

In case that the upper bound $W^{per}_0(s_0)$ and lower bound $W^{FO}_0(s_0)$ are close to each other, we can make sure that the performance of the percentile policy is close to the optimal policy. Thus, in Section \ref{sec:simulation}, we use numerical simulations of the ratio of $W^{per}_0(s_0)/W^{FO}_0(s_0)$ to evaluate the performance of our presented heuristic percentile policies.

\section{\textbf{Extensions of Policies for a Given Initial Belief}}\label{sec:belief}
In our formulation, we assume that the decision-maker is aware of the initial state of the system $s_0$. In many applications, the decision-maker might only have an estimate of how likely it is for the system to be in a given state. More precisely, let us assume that only a belief about the actual state is given (often it makes sense to pick this to be a uniform distribution over the state space to represent little bias towards any given state). Mathematically, we denote the initial belief with a vector $b_0=[b_0(0),...,b_0(M)]$, which represents the initial probability distribution of $B_0$ over all possible states in $\mathcal{M}$. Note that our decision starts from time step $1$.
Now, given the initial belief vector $b_0$, the goal is to minimize the total expected (discounted) cost over the horizon $T$, given by:
\begin{align}\nonumber
%\min_{\pi} J_T^{\pi}(s_0)=\min_{\pi} 
\mathbb{E}\{\sum_{t=1}^T {\beta ^{t-1}  C(B_t;r_t)}|b_0\}.
\end{align}
Next we show that a simple extension of our work in previous sections can be exploited in order to account for the initial belief (and lack of full initial knowledge). 

In this case, every policy has to account for the initial phase where the exact state of the system has never been observed. In particular, the optimal policy consists of an initial sequence that is computed based on the initial belief $b_0$. Thus, all policies need one more action sequence denoted by $\underline{\pi}_{b_0}$ with the length of $T$, beside all sequences of $\underline{\pi}_{s,t}$ for $s\in \mathcal{M}$ and $t=1,...,T-1$ presented in Section \ref{sec:seqDP}
and one more step of recursion is needed as follows:
\begin{align}\nonumber
W^{\pi}_{0}(b_0)&=\sum_{s=0}^M b_0(s) \times [ _{\underline{\pi}_{b_0}}\Gamma(s,t)\nonumber\\
&+\sum_{\tau=1}^T \beta^{\tau} \sum_{s'=0}^{\underline{\pi}_{b_0}(\tau)-1}\ _{\underline{\pi}_{b_0}}P(s,s',\tau) W^{\pi}_{\tau}(s')].
\end{align}
Similarly, given the decision-maker's initial belief, a percentile policy has to be appended by an initial phase. In this phase, the policy has to also include an initial sequence of actions prior to the first full observation of the state:
\begin{align}\nonumber%\label{eq:rptha2}
&\underline{\pi}^{per}_{b_1}(\tau)=\min \{r \in \mathcal{M}:\nonumber\\
&\frac{\sum_{i=0}^{r} [\sum_{s=0}^M b_0(s)\ _{\underline{\pi}_{b_0}}P(s,i,\tau)]}{\sum_{i=0}^{M} [\sum_{s=0}^M b_0(s)\_{\underline{\pi}_{b_0}}P(s,i,\tau)]} \geq h_{b_1}\}, \ \forall \tau=1,...,T.
\end{align}
where the corresponding threshold is denoted by $h_{b_1}$ in general and $h^{FRP}_{b_1}$ for FRP policy.

\section{\textbf{Numerical Results}}\label{sec:simulation}

In this section, first, we present two examples, one where the optimal policy and FRP policy meet and select the same action sequences, and second, where the optimal policy can not be a percentile policy, therefore FRP may select different action sequences from those selected by the optimal policy.
Next, we compare the performances of FRP and myopic policies for a large horizon and evaluate their performances compared to the optimal policy for a small horizon.

\subsection{\textbf{An Example where FRP Policy is Optimal}}

Here, we see an example with the same setting as Table \ref{fig:opttable} where FRP policy results in the same action sequences as the optimal policy with parameters $M=2$, $T=7$, $c_u=c_l=1$, $\beta=1$, and the transition matrix given in (\ref{eq:Pran}). The thresholds selected by FRP policy is shown in Table \ref{fig:hpttable} for this example. 
Thus, the action sequences $\underline{\pi}^{FRP}_{s,t}$ generated by the given thresholds are equivalent to the action sequences $\underline{\pi}^{opt}_{s,t}$ given in Table \ref{fig:hpttable}.
Note that the threshold selected by the myopic policy is $h^m=0.5$.

\begin{table}[b]
	\caption{The thresholds $h^{FRP}_{s,t}$ selected by FRP policy for each full observed state $s$ at time $t$, for the same setting as Table \ref{fig:opttable}.}
    \centerline{\includegraphics[width=.8\linewidth]{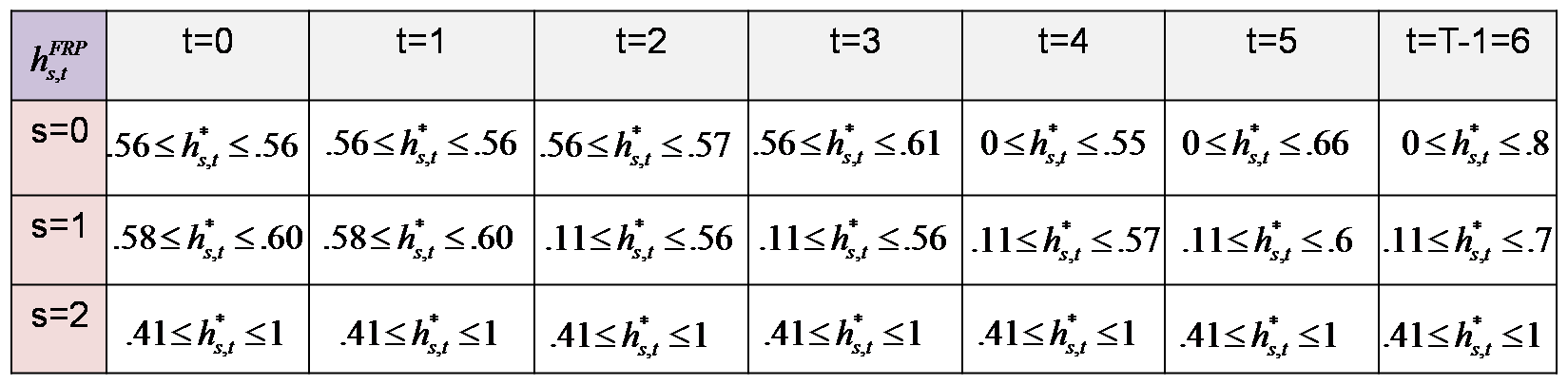}}    
    \label{fig:hpttable}
\end{table}

Note that each action sequence of optimal policy $\underline{\pi}_{s,t}^{opt}$ given in Table \ref{fig:opttable} can be regenerated by applying time-varying non-unique thresholds, $h^{opt}_{s,t}(\tau), \tau=t+1,...,T$. The valid values for each threshold $h^{opt}_{s,t}(\tau)$ could be shown by a range of real numbers between 0 and 1.
The FRP policy tries to select the threshold $h^{FRP}_{s,t}$ inside the intersection of the ranges of $h^{opt}_{s,t}(\tau), \tau=t+1,...,T$. For instance, Fig. \ref{fig:hptopt} shows the range of valid thresholds for the observation pairs $(s,t)$ of $(0,0)$, $(1,0)$ and $(1,3)$. For these pairs of observation, FRP policy could select the threshold $h^{FRP}_{s,t}$ from the intersection of the ranges of valid thresholds $h^{opt}_{s,t}(\tau), \tau=t+1,...,T$ of the optimal policy. 

\begin{figure*}
    \centerline{\includegraphics[width=\linewidth]{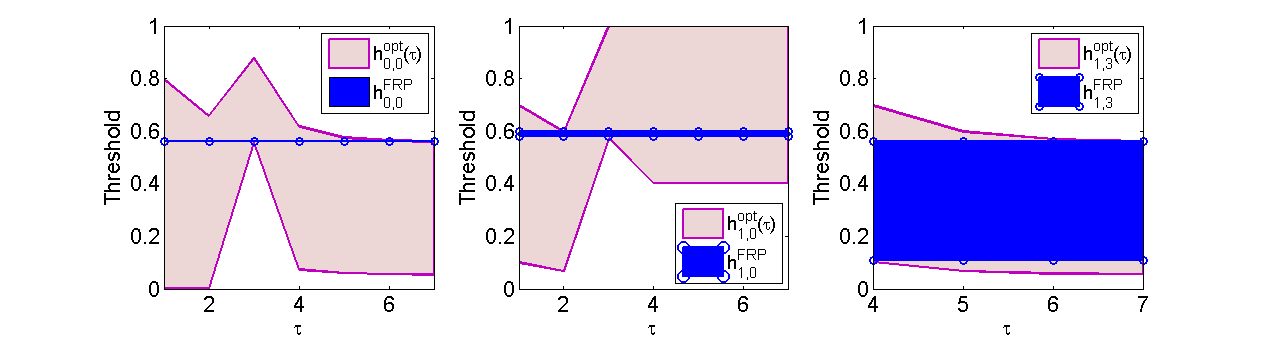}}
    \caption{The range of valid thresholds of the optimal policy $h^{opt}_{s,t}(\tau)$ and the threshold of FRP policy $h^{FRP}_{s,t}$ corresponding to the full observation of $s,t$, versus time step $\tau$, for the same setting as Table \ref{fig:opttable}.}
    \label{fig:hptopt}
\end{figure*}
%\newpage~\newpage

\subsection{\textbf{An Example where Optimal Policy is not a Percentile Policy}}

Here, we see an example where the optimal policy can not be a percentile policy and its action sequences are different from those selected by FRP policy. The parameters are assumed to be $M=2$, $T=7$, $c_u=c_l=1$, $\beta=1$, and transition matrix is given by:
\begin{align}\footnotesize
P=\begin{bmatrix}
    .9 & .1 & 0  \\
    .1 & .8 & .1 \\
     0 & .1 & .9 \\
\end{bmatrix}.\label{eq:Pe01}
\end{align}

The action sequences $\underline{\pi}^{FRP}_{s,t}$ generated by FRP policy and the action sequences $\underline{\pi}^{opt}_{s,t}$ generated by the optimal policy are given in Table \ref{fig:aoptpt_table}.
The action sequences are the same for the optimal and FRP policies except in one case, where $s=0$ and $t=0$, \textit{i.e.} $\underline{\pi}^{opt}_{0,0}(\tau)=1$, $\underline{\pi}^{FRP}_{0,0}(\tau)=2$ for $\tau=6,7$.
The reason for $\underline{\pi}^{opt}_{0,0}\neq \underline{\pi}^{FRP}_{0,0}$, is that the action sequence $\underline{\pi}^{opt}_{0,0}$ is not following the percentile structure. In other words, we can not find a fixed threshold to reconstruct this action sequence.

\begin{table}[t]
    \caption{The action sequences $\underline{\pi}^{FRP}_{s,t}$ and $\underline{\pi}^{opt}_{s,t}$ selected by FRP and optimal policies for each full observed state $s$ at time $t$, for $M=2$, $T=7$, $c_u=c_l=1$, $\beta=1$ and $P$ matrix given in (\ref{eq:Pe01}). The sequences are similar for all $(s,t)$ except $(0,0)$. The FRP actions different from optimal actions are denoted by red color.}
    \centerline{\includegraphics[width=.8\linewidth]{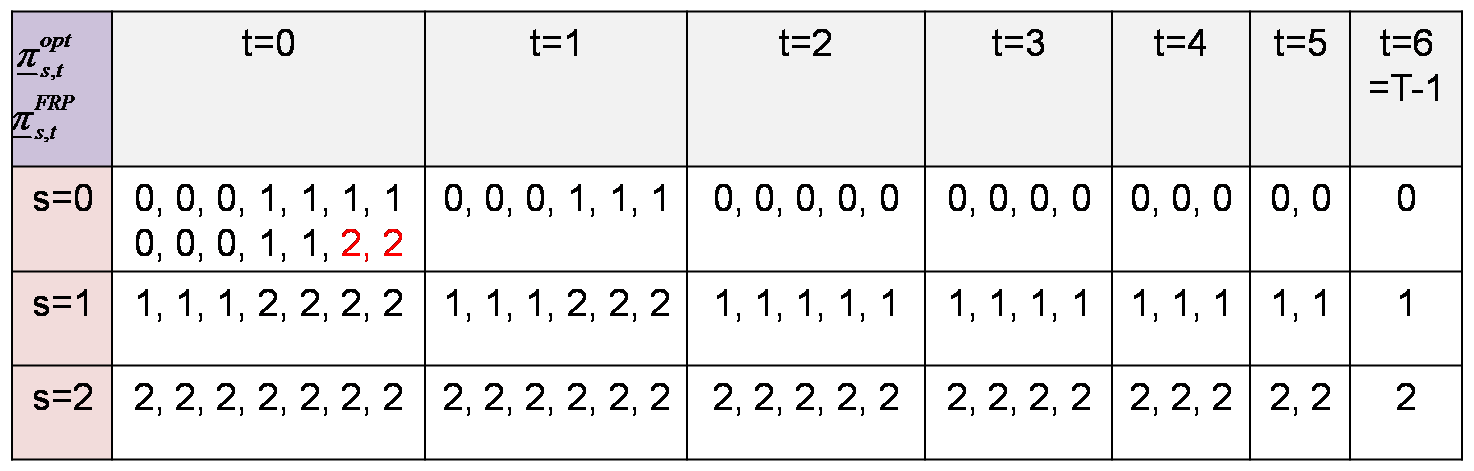}}
    \label{fig:aoptpt_table}
\end{table}

As Fig. \ref{fig:hh00_e01} shows, for the observation pair of $s=0$ and $t=0$, the intersection of the ranges of valid thresholds, $h^{opt}_{s,t}(\tau), \tau=t+1,...,T$ are empty. The threshold selected by FRP policy is higher than the thresholds of the optimal policy for $\tau\geq 6$. This confirms that the optimal policy is not necessarily a percentile policy.

\begin{figure}[b]
    \centerline{\includegraphics[width=.9\linewidth]{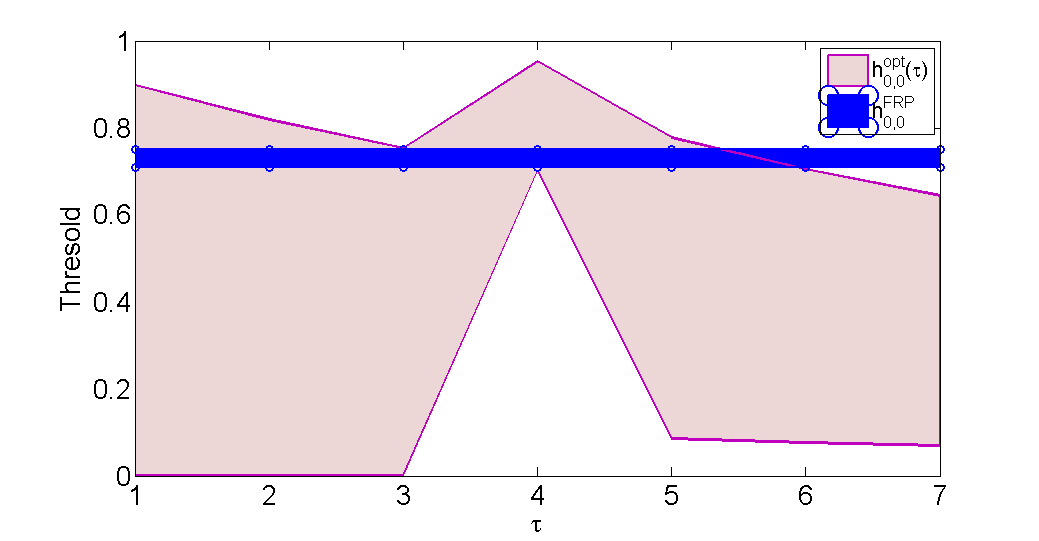}}
    \caption{The range of valid thresholds of the optimal policy $h^{opt}_{s,t}(\tau)$ and the threshold of FRP policy $h^{FRP}_{s,t}$ corresponding to the full observation of $s=0$ at time $t=0$, versus time step $\tau$, for the same setting as Table \ref{fig:aoptpt_table}.}
    \label{fig:hh00_e01}
\end{figure}

\subsection{\textbf{Evaluation of Percentile Policies}}

We now present some numerical results to evaluate the performance of the introduced FRP policy.
The simulation parameters, except in the figures where their effect is considered, are fixed as follows: the highest state $M =4$, the under-utilization cost coefficient $c_l=1$, the discount factor $\beta=1$, the horizon $T=7$ whenever we compute the optimal policy and $T=30$ otherwise. We consider the following transition matrix:
\begin{align}\footnotesize
P=\begin{bmatrix}
    1-\epsilon & \epsilon & 0 & & \dots  & 0 \\
    \epsilon & 1-2\epsilon & \epsilon & & \dots & 0 \\
   \vdots & \vdots &  & \ddots & \vdots \\
    0 & \dots & & \epsilon & 1-2\epsilon & \epsilon \\
    0 & \dots &  & 0 & \epsilon & 1-\epsilon
\end{bmatrix},\label{eq:Peps}
\end{align}
with the size of $(M+1)\times (M+1)$, for given $\epsilon=0.3$.
Note that the results achieved by increasing $c_u$ for a fixed $c_l$ is equivalent to those achieved by decreasing $c_l$ for a fixed $c_u$ and vice verse.
We consider two initial cases: (i) Full observation state $s_0=0$, (ii) Given uniform initial belief $b_0$.
%We compare the performance of FRP policy with the myopic policy for a large horizon and with the optimal policy for a small horizon.

%\subsection{\textbf{Thresholds of FRP Policy}}

Figures \ref{fig:hptbeta} and \ref{fig:hptcu} show the best thresholds selected by FRP policy for initial full observation of $s_0=0$ and the uniform initial belief $b_0$ versus the discount factor $\beta$ (for $c_u=5c_l$) and versus $c_u$ (for $\beta=1$), respectively.

\begin{figure}[b]
    \centerline{\includegraphics[width=.9\linewidth]{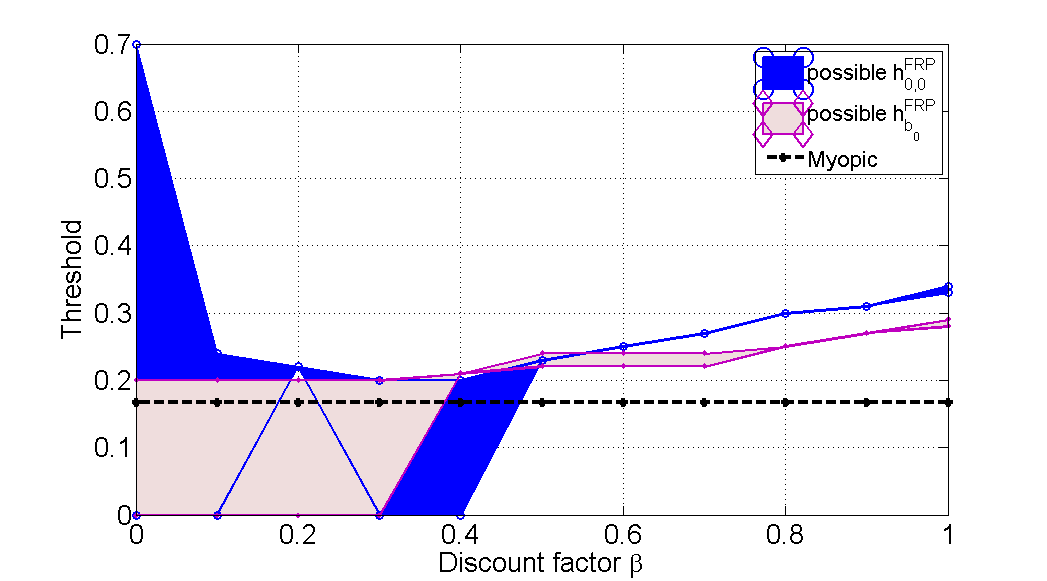}}
    \caption{The thresholds versus the discount factor $\beta$, for $M=4$, $c_l=1$, $c_u=5$ and $\epsilon=0.3$.}
    \label{fig:hptbeta}
\end{figure}

\begin{figure}[t]
    \centerline{\includegraphics[width=.9\linewidth]{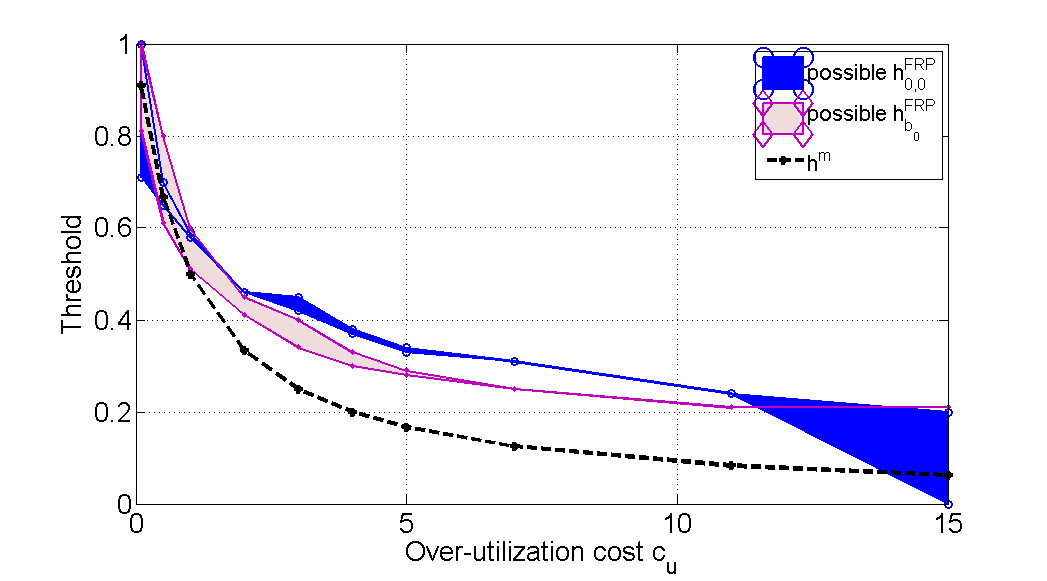}}
    \caption{The thresholds versus the over-utilization cost $c_u$, for $M=4$, $c_l=1$, $\beta=1$ and $\epsilon=0.3$.}
    \label{fig:hptcu}
\end{figure}

The thresholds are not unique and could be any values in the corresponding colored ranges. %We draw the ranges for possible thresholds for two different initial cases, $s_0=0$ and uniform $b_1$.  
As shown in Fig. \ref{fig:hptbeta}, for small values of $\beta$ the threshold of the myopic policy, $h^m=c_l/(c_u+c_l)=0.1667$ is a valid threshold for FRP policy. But for larger values of $\beta$ the myopic policy can not be a good heuristic, since the future is more important. The threshold of FRP policy increases with $\beta$ to increase the chance of full observation which could be useful for future decisions.  
The thresholds decrease with increasing $c_u$, as is obvious from Fig. \ref{fig:hptcu}, since for larger values of $c_u$ the policies prefer to behave more conservatively, choosing smaller action.
For larger values of $c_u$ the thresholds of FRP policy are higher than those chosen by the myopic policy. This difference is more obvious for $h_{b_0}$.

The best thresholds of FRP policy for initial full observation of $s_0=0$ and the uniform initial belief $b_0$ versus the horizon $T$ are given in Fig. \ref{fig:hptT}. As shown in the figure, for larger horizons, the selected thresholds are higher than the threshold of the myopic policy $h^m$ in order to increase the chance of full observation useful for future decisions.

\begin{figure}[b]
    \centerline{\includegraphics[width=.9\linewidth]{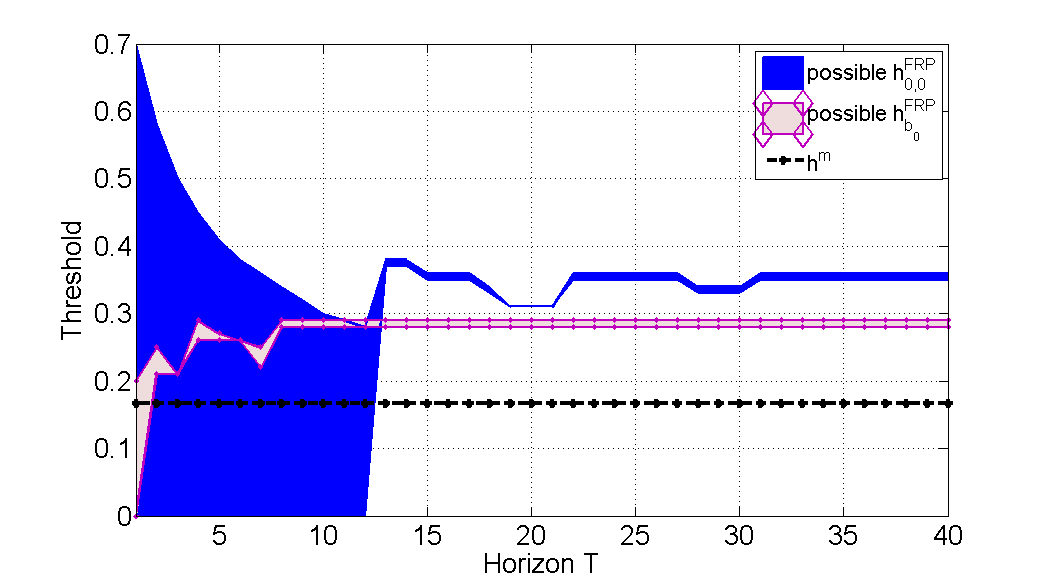}}
    \caption{The thresholds versus the horizon $T$, for $M=4$, $\beta=1$, $c_l=1$, $c_u=5$ and $\epsilon=0.3$.}
    \label{fig:hptT}
\end{figure}

Figure \ref{fig:hpteps} considers the effect of transition matrix $P$ on the thresholds selected by FRP policy. To see its effect we plot the threshold versus parameter $\epsilon$ of $P$ matrix given in (\ref{eq:Peps}). Smaller $\epsilon$ means that $P$ is closer to identity matrix (static process) where full observation could reveal more information about the future and decrease the future cost a lot. Thus the difference between the selected thresholds and the threshold of the myopic policy are more significant (specially for $b_0$).

\begin{figure}[t]
    \centerline{\includegraphics[width=.9\linewidth]{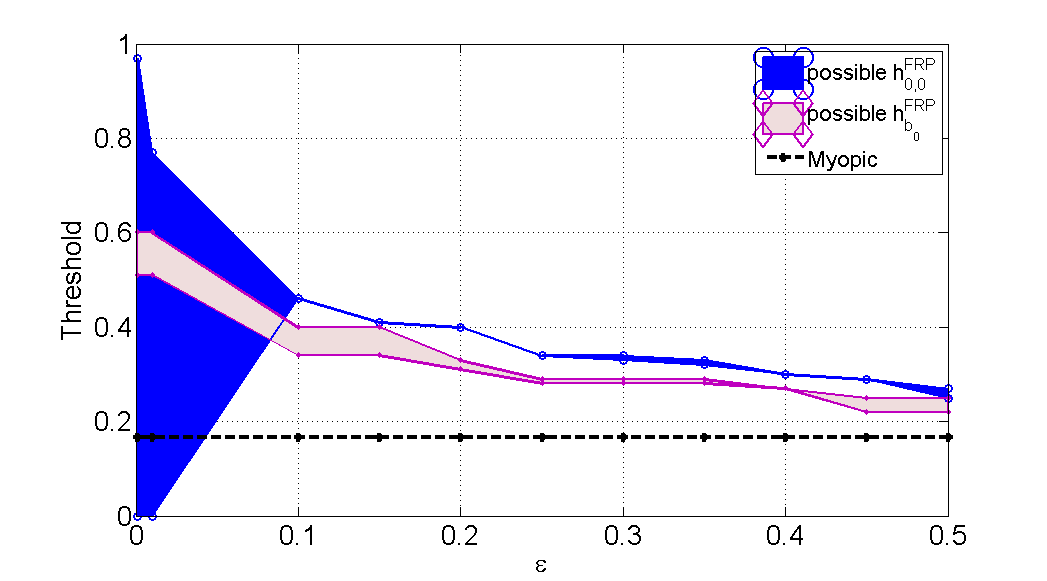}}
    \caption{The thresholds versus $\epsilon$, for $M=4$, $c_u=5$, $c_l=1$, $\beta=1$, $T=7$.}
    \label{fig:hpteps}
\end{figure}

\subsection{\textbf{Performance of Percentile Policies}}

We now present the numerical results to evaluate the performances of the percentile policies. For this evaluation we use the ratio of the total expected cost of the percentile policies to the expected cost of FO genie (as an upper bound on the ratio of their expected cost to the expected cost of the optimal policy). Figure \ref{fig:perfbeta} shows the cost ratios versus the discount factor, $\beta$, for $c_u=5c_l$. The percentile policies have a better performance for smaller values of $\beta$ where the myopic policy also performs well, but for larger values of discount factor since the future matters more, FRP policy outperforms the myopic policy. Interestingly, FRP policy has a cost ratio less than $1.7$ for all values of $\beta$, for the given parameters, \textit{i.e.} the cost of FRP policy is not more than $1.7$ times the cost of the optimal policy.

\begin{figure}[b]
    \centerline{\includegraphics[width=.9\linewidth]{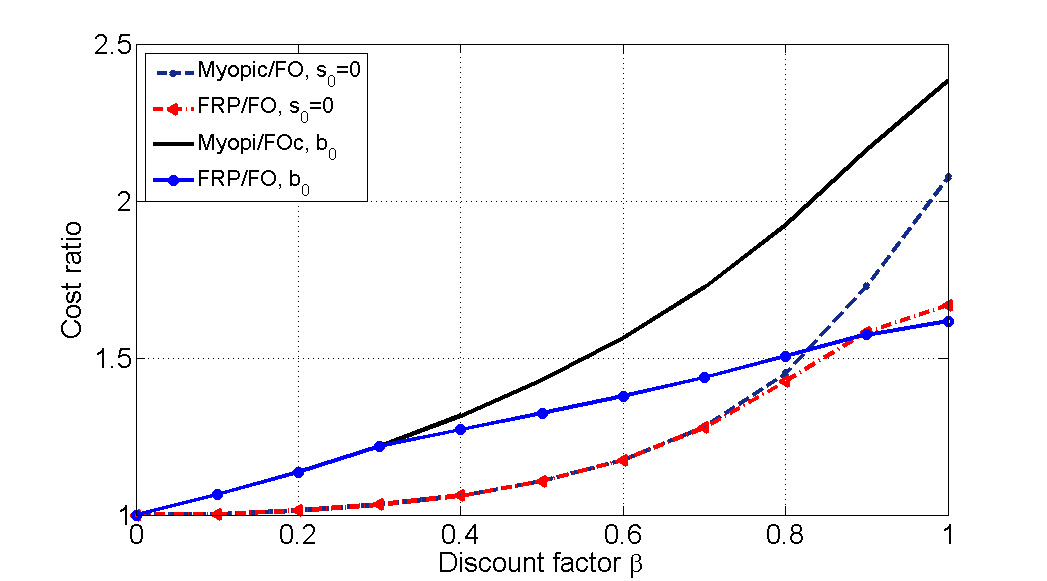}}
    \caption{The cost ratios versus the discount factor $\beta$, for $M=4$, $c_l=1$, $c_u=5$, $T=30$ and $\epsilon=0.3$.}
    \label{fig:perfbeta}
\end{figure}

Figure \ref{fig:perfcu} shows the cost ratios versus the over-utilization cost, $c_u$, for $c_l=1$. The percentile policies have a better performance for smaller values of $c_u$ and FRP policy outperforms the myopic policy for medium values of $c_u$. For both large/small values of $c_u$, the performances of the myopic and FRP policies are close, since both of them behave as conservative/aggressive as the optimal policy.

\begin{figure}[t]
    \centerline{\includegraphics[width=.9\linewidth]{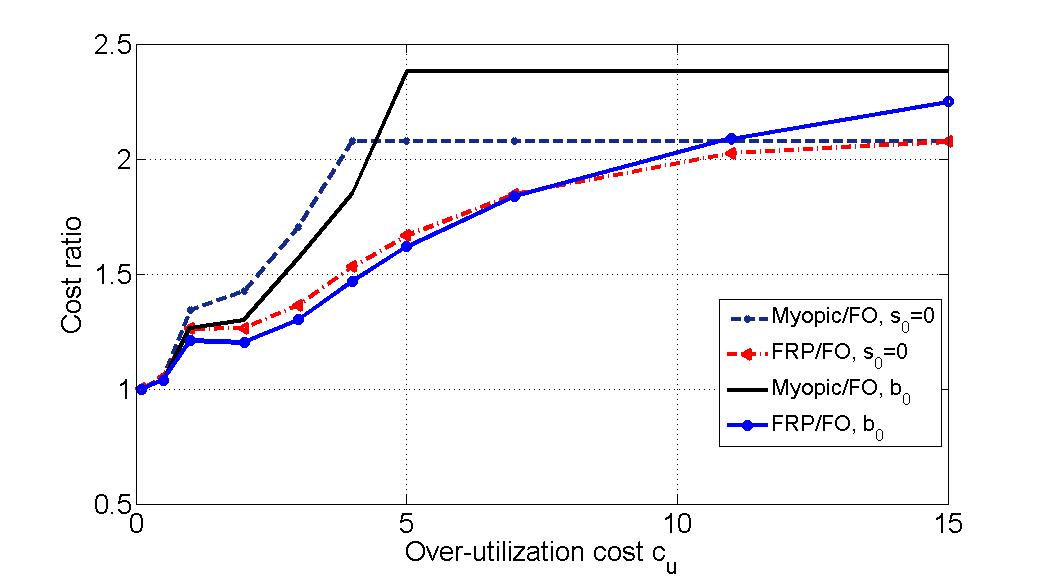}}
    \caption{The cost ratios versus the over-utilization cost, $c_u$, for $M=4$, $c_l=1$, $\beta=1$, $T=30$ and $\epsilon=0.3$.}
    \label{fig:perfcu}
\end{figure}

Figure \ref{fig:perfT} shows the cost ratios versus versus the horizon $T$. It is shown that the cost ratios of the percentile policies are increasing by $T$.
Figure \ref{fig:perfesp} considers the effect of transition matrix $P$ on the performances of the myopic and FRP policies. For smaller values of $\epsilon$, the $P$ matrix given in (\ref{eq:Peps}) is closer to the identity matrix where the ratio to the cost of FO genie policy may not be good for the performance evaluation (the ratio to the cost of the optimal policy might be much better).

\begin{figure}[htb]
    \centerline{\includegraphics[width=.9\linewidth]{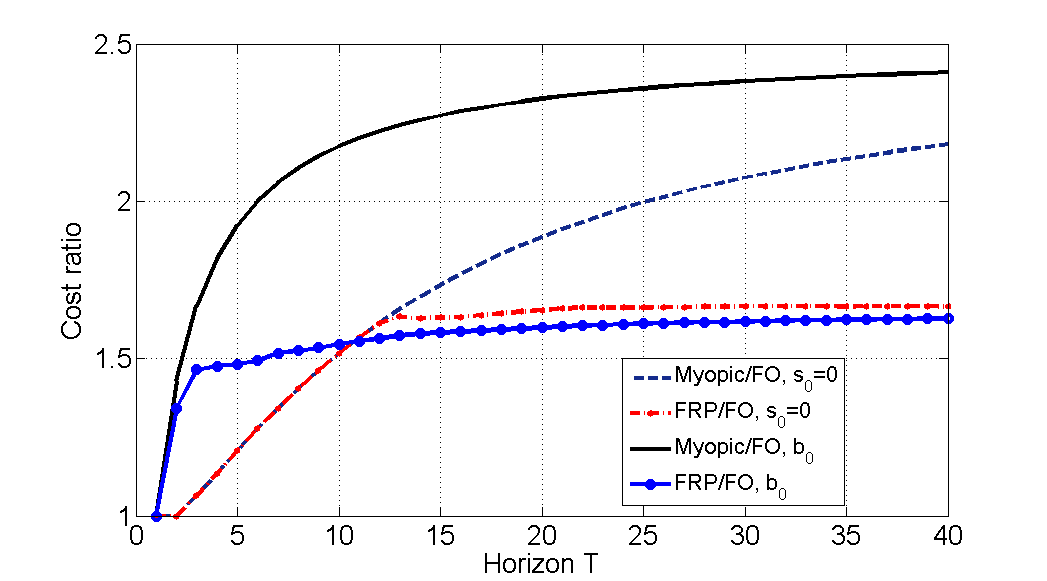}}
    \caption{The cost ratios versus the horizon $T$, for $M=4$, $\beta=1$, $c_l=1$, $c_u=5$ and $\epsilon=0.3$.}
    \label{fig:perfT}
\end{figure}

\begin{figure}[htb]
    \centerline{\includegraphics[width=.9\linewidth]{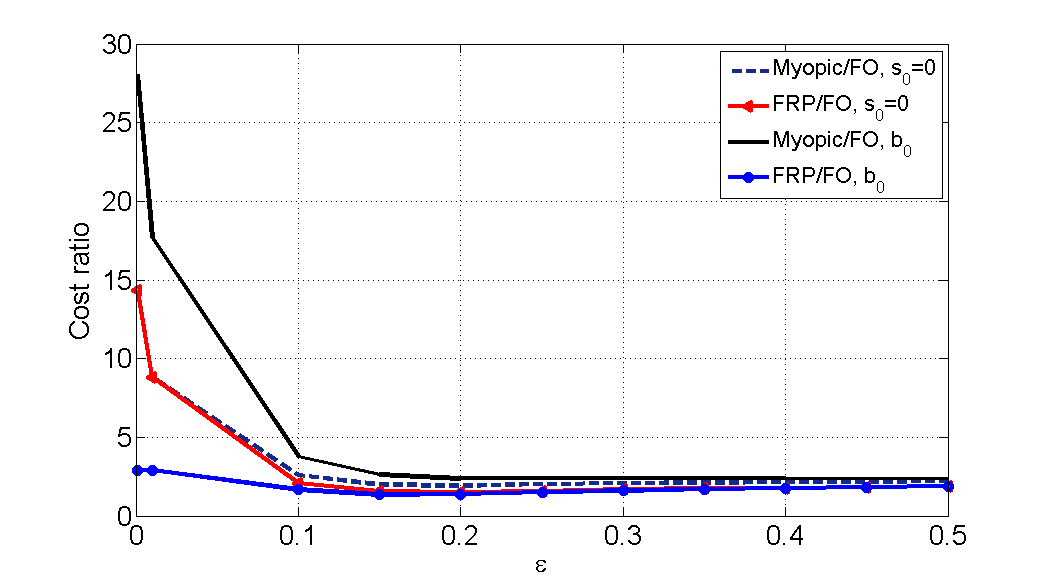}}
    \caption{The cost ratios versus $\epsilon$, for $M=4$, $c_u=5$, $c_l=1$, $\beta=1$, $T=30$.}
    \label{fig:perfesp}
\end{figure}

To see the performances for a larger state set size, Fig. \ref{fig:perfcuM20} shows the cost ratios of policies versus the over-utilization cost $c_u$ for $M=19$ and the transition matrix of 
\begin{align}\footnotesize
P=\begin{bmatrix}
    .6 & .1 & .2 & .1 & 0 & &  &\dots  & 0 \\
     .4 & .2 & .1 & .2 & .1 & 0 & & \dots  & 0 \\
     .3 & .1 & .2 & .1 & .2 & .1 & 0 & \dots  & 0 \\
    \vdots & \vdots &  & \ddots & \vdots \\
    0 & \dots  & 0 & .3 & .1 & .2 & .1 & .2 & .1 \\
    0 & \dots  &  & 0 & .3 & .1 & .2 & .1 & .3 \\
    0 & \dots  & & & 0 & .3 & .1 & .2 & .4
\end{bmatrix},\label{eq:P312121}
\end{align}
with the size of $20\times 20$. 
Comparing this figure to Fig. \ref{fig:perfcu}, we can conclude that the cost ratios are larger for a larger state set size.
%Note that the performances for the uniform initial belief and the initial full observation are closer to each other when the size of the state set is large.

\begin{figure}[htb]
    \centerline{\includegraphics[width=.9\linewidth]{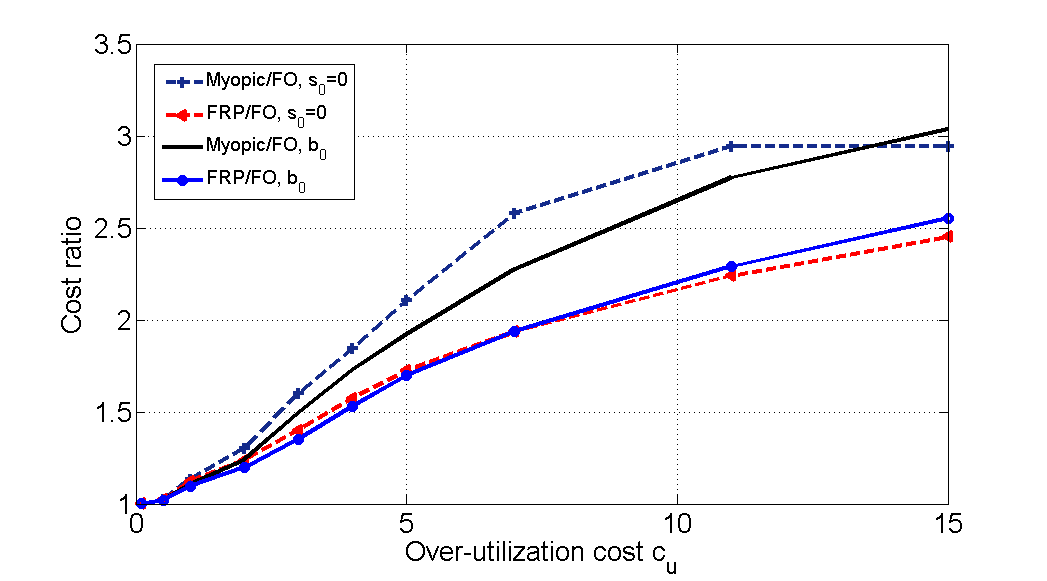}}
    \caption{The cost ratios versus the over-utilization cost, $c_u$, for $M=19$, $c_l=1$, $\beta=1$, $T=30$.}
    \label{fig:perfcuM20}
\end{figure}

\subsection{\textbf{Comparison with the Optimal Policy}}

To compare the performances of the percentile policies with the optimal policy, we consider a small horizon of $T=7$ due to the complexity of DP. 
Figure \ref{fig:cbeta} shows that the expected cost of FRP and optimal policies are equal and they outperform the myopic policy (even though the ratios to the cost of FO genie policy is $1.35$ for $\beta=1$).
Figure \ref{fig:ccu} shows the ratio of the total expected cost of the policies versus $c_u$ for $T=7$. This figure shows that the performances of FRP policy and optimal policies are equal for larger values of $c_u$ (cost ratio equals to 1). And also the cost of FO policy that we used to evaluate the performances of the percentile policies in the previous subsection may provide a looser bound on the actual performances, for larger values of $c_u$.

\begin{figure}[htb]
    \centerline{\includegraphics[width=.9\linewidth]{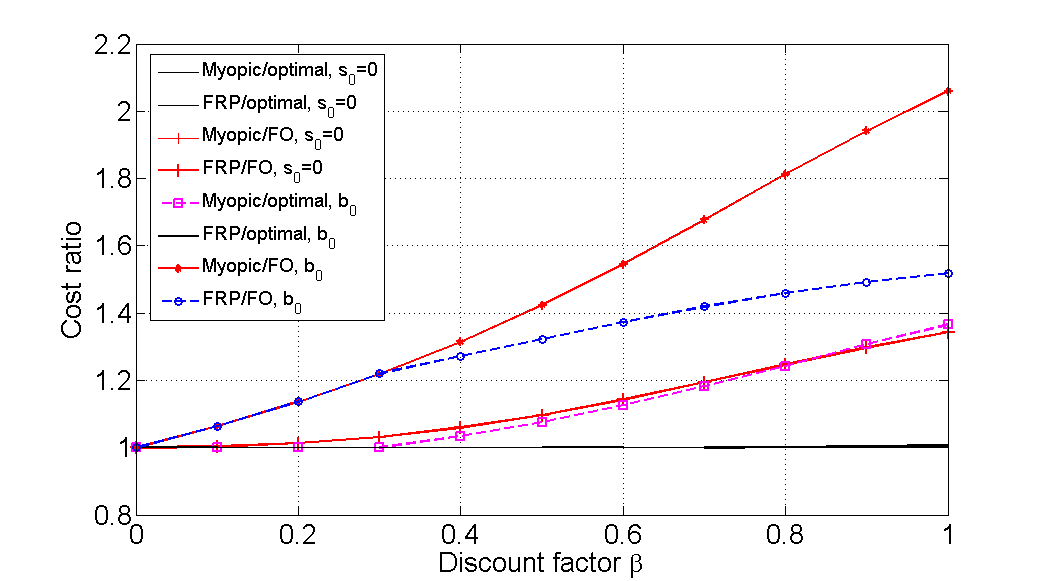}}
    \caption{The total expected cost versus the discount factor $\beta$, for $M=4$, $c_l=1$, $c_u=5$, $\epsilon=0.3$, $T=7$.}
    \label{fig:cbeta}
\end{figure}

\begin{figure}[htb]
    \centerline{\includegraphics[width=.9\linewidth]{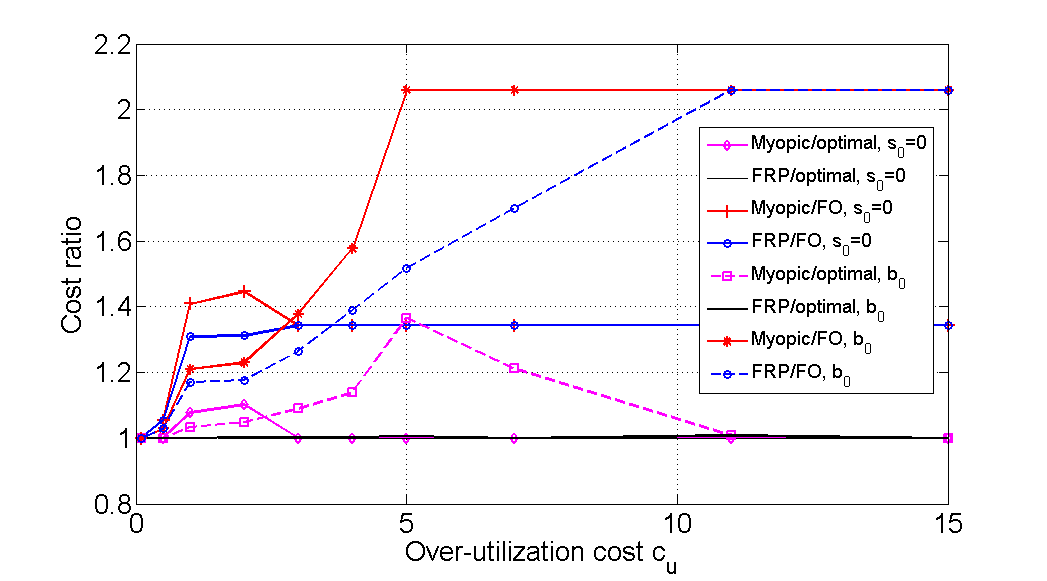}}
    \caption{The total expected cost versus the over-utilization cost, $c_u$, for $M=4$, $c_l=1$, $\beta=1$, $\epsilon=0.3$, $T=7$.}
    \label{fig:ccu}
\end{figure}

Figure \ref{fig:cesp} considers the effect of $\epsilon$ (parameter in the transition matrix $P$) on the ratio of the expected cost of the policies. For small values of $\epsilon$ both FRP and myopic policies under-perform the optimal policy and the performance bound provided by FO policy is looser.

\begin{figure}[t]
    \centerline{\includegraphics[width=.9\linewidth]{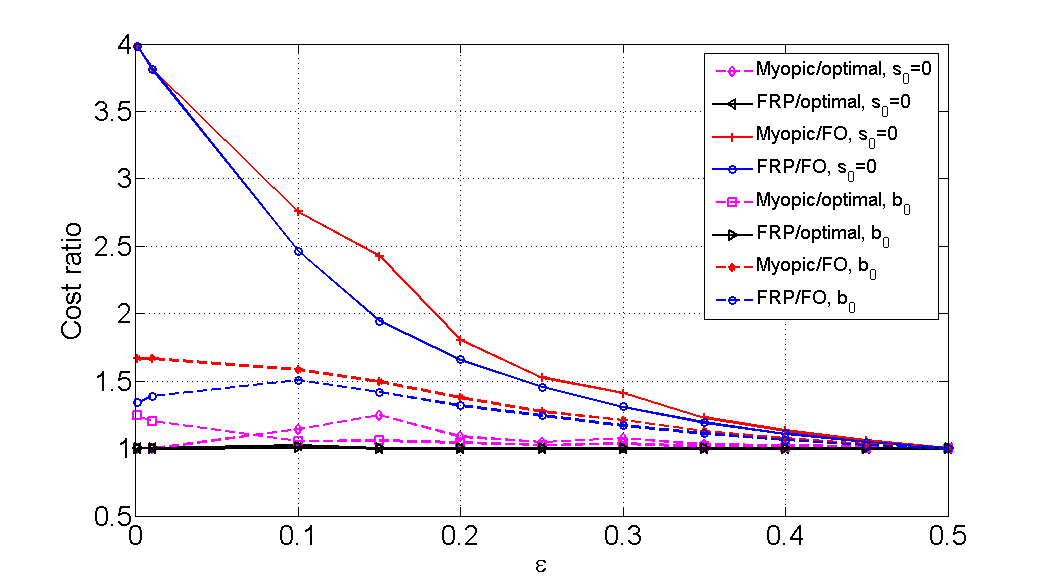}}
    \caption{The total expected cost versus $\epsilon$, for $M=4$, $c_u=1$, $c_l=1$, $\beta=1$, $T=7$.}
    \label{fig:cesp}
\end{figure}

\section{\textbf{Summary and Conclusion} } \label{sec:conclusion}
We have considered the tracking problem of Markovian random processes in which
the goal is to select the best action sequences starting from any full observation in order to minimize the total expected (discounted) cost accumulated over a finite horizon. We have modeled this decision-making problem as a POMDP in a novel formulation based on the last observed state and the observation time.

We have shown that the complexity of computing the optimal policy for general transition matrices of the background Markovian process is exponentially high and thus numerically solving for the optimal policy is limited to very small horizons. 
Instead, we have introduced a new class of policies with a percentile threshold structure, which have been called percentile policies, and evaluated their performances. We have introduced a novel heuristic policy which search among the possible thresholds to find the best thresholds, for each full observed state and the observation time, minimizing the total expected cost of the percentile policy. Numerical results have shown that this policy outperforms the myopic policy. Furthermore, for small horizons we have compared this heuristic policy with the optimal policy and have shown that its performance is very close to the optimal policy.

As future works, we aim to identify the conditions where FRP policy is optimal.
To compute the heuristic policy we have used trial and error method to find the best possible thresholds among a finite threshold set for each full observation time and state. We will analyze the cost functions of percentile policies to find the optimal threshold in the ranges of $[0,1]$.
We will also consider an inventory management problem where the leftover inventories can be carried over to the next time step and derive the percentile policy and its performance guarantee. We also consider a more complicated scenario where the transition matrix is unknown and needs to be learned over time.

\appendices
\section{}\label{app:comp}

\begin{proof}[\textbf{Proof of Proposition \ref{prop:compPT}}]
To compute the action sequences and the expected cost of the percentile policy with given thresholds we need to compute $W^{per}_0(s)$ and $W^{per}_0(b_0)$ of percentile policy for all $s\in\mathcal{M}$ and $t=0,...,T-1$. 
The complexity of functions are as follows:
\begin{itemize}
\item $\underline{\pi}^{per}_{s,t}(t')$ for any $t'$ and given $s,t$, $h_{s,t}$ and transitions: $\Theta(M+1)$,
\item To compute $_{\underline{a}}P(s_1,s_2,\tau)$ for all $s_1,s_2,\tau$ and action sequences with the length of $T-t$: $\Theta((T-t)(M+1)^2)$,
\item $\Gamma(s,t)$ for given parameters above: $\Theta((T-t)\times (M+1))$,
\end{itemize}

In total, given $W^{per}_{t'}(s')$ for all $s'$ and $t'=t+1,...,T-1$, to compute $W^{per}_t(s)$ for a given threshold and $s$, we need $\Theta((T-t)(M+1)^2)$ computations. Thus for all $s\in \mathcal{M}$ we need $(M+1)\times \Theta((T-t)(M+1)^2)$ computations.

Now to compute all action sequences and expected costs, we use backward recursion on time from horizon $T-1$ to $0$.
 Thus, for $t=T-1$ the complexity is simply $\Theta(M+1)$. We use the already computed values at time $t+1,...,T-1$ to compute the minimum cost and the action sequences at time $t$. 
Thus, the total complexity of percentile policy equals to $\Theta(\sum_{t=0}^{T-1} (T-t)\times(M+1)^3)=\Theta(T^2 (M+1)^3)$.
\end{proof}

\section{}\label{app_lproofs}

\begin{proof}[\textbf{Proof of Proposition \ref{C_mybb1}}]
The proof is straightforward using the uni-modality of expected immediate cost given in the following lemma.
\begin{lem} \label{Myopic_cor}
The expected immediate cost for any belief $b$ is a uni-modal function of the action $r$, where a function is called uni-modal iff it has at most one minimum and being decreasing/increasing before/after the minimum.
\end{lem}

Thus the expected immediate cost is decreasing with respect to $r$ when $\sum_{i=m}^{r} {b(i)} < \frac{c_l}{c_l+c_u}$ and it is increasing otherwise. Therefore, the myopic action which minimizes the expected immediate cost, given in (\ref{eq:myopic}), is equal to the lowest action $r$ satisfying the inequality $\sum_{i=m}^{r} {b(i)} \geq \frac{c_l}{c_l+c_u}$.
\end{proof}
%\medskip

\begin{proof}[\textbf{Proof of Lemma \ref{Myopic_cor}}]
To prove this lemma, let us compute the derivative of the expected immediate cost for the belief vector $b$ with respect to action $r$ given in (\ref{eq:vfnR}),
\begin{align}
\Delta \bar C(b;r)&=\bar C(b;r+1)-\bar C(b;r) =(c_l+c_u) \sum_{i=m}^{r} {b(i)}-c_l.\nonumber% \label{eq:c2}
\end{align}
It's easily seen that if the inequality in (\ref{eq:rmc1}) holds, $\Delta \bar C(b;r)\geq 0$. Otherwise it is negative. This concludes that $\bar C(b;r)$ is uni-modal with unique minimum at the myopic action given in (\ref{eq:rmc1}).
\end{proof}
%\medskip

\section{}\label{sec:special}

\textbf{Optimal Policy for Special Cases:}
Even though computing the optimal policy is intractable in general, there are some trivial cases in which the optimal policy is achievable. In the following subsections, we present simple special cases of zero under/over utilization costs, %static/drift processes (special $P$ matrices), 
and cases where the myopic policy is optimal.

\subsection{\textbf{Zero Under-utilization or Over-utilization Cost}}

For two special cases the optimal policy is trivial. First if the under-utilization cost coefficient is zero, $c_l=0$, the optimal policy selects the lowest actions $0$ all the time or any action sequences in which the transition probabilities to the actions below them are zero, \textit{i.e.}, for any full observation of state $s$ at time $t$, $\underline{\pi}^{opt}_{s,t}(\tau) \leq \min \{i: P^{\tau-t}_{s,i}\neq 0\}$ for $\tau=t+1,...T$, $s\in\mathcal{M}$.
Second, if the over-utilization cost coefficient is zero, $c_u=0$, the optimal policy selects the highest action $1$ all the time or any action sequences in which the transition probabilities to the actions above them are zero, \textit{i.e.}, $\underline{\pi}^{opt}_{s,t}(\tau) \geq \max \{i: P^{\tau-t}_{s,i}\neq 0\}$ for $\tau=t+1,...T$, $s\in\mathcal{M}$. 
Note that the total expected (discounted) costs achieved by the optimal policy for these two trivial cases are always zero. 
The optimal policy for zero under/over-utilization cost is also a percentile policy with a threshold of zero/one.

\subsection{\textbf{Cases where Myopic Policy is Optimal}}
There are some cases where the optimal policy is equivalent to the myopic policy, given by:
\begin{itemize}
\item Zero discount factor ($\beta=0$),
\item Independent and Identically Distributed (i.i.d.) Processes,
\item When there is no ambiguity about the future, \textit{i.e.} the rows of the transition matrix are equal to a fixed unit vector. $P_{i,.}=I_{a}, \forall i\in \mathcal{M}$ for a given $a\in \mathcal{M}$ where $I_a$ is a $M+1$-dimensional unit vector with 1 in the $a$-th position and zero elsewhere.
\end{itemize}

Thus the myopic policy is a percentile policy with the threshold of $h^{m}=\frac{c_l}{c_l+c_u}$.

\subsection{\textbf{Any transition matrix with unit vectors as rows}}
Here we consider the cases where each row of the matrix is a unit vector, \textit{i.e.} there is only one 1 in any row and the rest of the elements are zero, \textit{e.g.} a static process with matrix $P=I$. For these cases, after any full observation
of state $s\in \mathcal{M}$ at time $t=0,...,T-1$, the actual state is simply known for the rest of the times with probability one. Therefore the optimal policy selects the actions equal to the actual states $\underline{\pi}^{opt}_{s,t}(\tau)=\{i: P^{\tau-t}(i)=1\}$ for $\tau=t+1,...,T$ and the expected cost-to-go is zero.

\section{}\label{app_beliefDP}

We could alternatively represent the decision problem based on the decision-maker's belief, \textit{i.e.} his posterior probability on the hidden state $B_t$ at each time step conditioned on past actions and observations (\cite{mansourifard2013bayesian, mansourifard2015, bensoussan2007multiperiod}).
Let the conditioned probability distribution of the state, given all past observations, is denoted by a belief vector $b_t=[b_t(0),...,b_t(M)]$, with elements of $b_t(k)=Pr(B_t=k|~past~observations), k \in \mathcal{M}$. In other words, $b_t$ represents the probability distribution of $B_t$ over all possible states in $\mathcal{M}$.
In this representation, the goal is to make a decision at each time step based on the history of observations; but due to the lack of full information, the decision-maker may only make the decision based on the belief vector. 
It can be shown that the belief vector is a sufficient statistic of the complete observation history (see \textit{e.g.}, \cite{smallwood1973optimal}). Note that to find the optimal action at each time, we need to know the history of actions and observations. The belief vector is a function of these parameters and is updated every time step based on the selected action and the observation, since all the actions and observations could effect the probability distribution of the states. Thus, all the history is compressed in the belief vector.

Here, we present the belief-based DP which results in the same optimal policy as the sequence-based DP presented in the main text,
as the following recursive equations: 
\begin{subequations}
\begin{align}
&V_t(b_t):=\min_{r_t} V_t(b_t;r_t), \label{eq:vfnmin}\\
&V_T(b_T;r_T)=\bar C_{T}(b_T;r_T),\label{eq:vfnT}\\
&V_t(b_t;r_t):=\bar C(b_t;r_t)+\beta \mathbb{E}\{V_{t+1}(b_{t+1})|r_t,b_t\},\ \ t<T  \label{eq:vfn}
\end{align}
\end{subequations}
where $b_{t+1}$ is the updated belief vector. It can be computed given the action $r_t$ and observations.
The value function $V_t(b_t)$  is the minimum expected cost-to-go when the current belief vector is $b_t$. 
Note that $V_t(b_t;r)$ is the expected cost-to-go after time $t$ under belief $b_t$ and action $r$ at time $t$ and following the optimal policy for time $t+1$ onward, with updated belief vector according to the action $r$.

The belief updating maps current belief vector, selected action, and the observation to the belief vector for the next time step:
\begin{align}\label{eq:pt1}
b_{t+1} = \begin{cases}
  T_{r_t}[b_{t}] P  &\text{if}\ B_t\geq r_t, \\
 I_{i}P  &\text{if}\  B_t=i<r_t,\\
\end{cases}\
\end{align}
where $I_{i}$ is the $M+1$-dimensional unit vector with 1 in the $i$-th position and 0 otherwise. Note that $I_{i}P$ is equivalent to the $i$-th row of matrix P, \textit{i.e.} $P_{i,.}$.
$T_{r}$ is a non-linear operation on a belief vector $b$, as follows:
\begin{align}\label{eq:Trp}
T_{r}[b](i) = \begin{cases}
  0  &\text{if}\ i < r, \\
 \frac{b(i)}{\sum_{j=r}^M b(j)}  &\text{if}\  i\geq r.\\
\end{cases}\
\end{align}
The update of the belief vector is derived in two steps. First step: when we get full observation about the actual state $i$ with taking action $r_t$, \textit{i.e.} when $i<r_t$, the probability distribution of the state could be updated to $I_i$, since we are sure that the probability of actual state being different than $i$ is zero.
And when we get partial observation that the actual state is higher than $r_t$, we can update the probability distribution as follows: force the probability of the actual state being less than $r_t$ to be zero and normalize the rest of the probabilities to sum up to one. Second step: the above updating corresponds to the probability distribution at the same time step of taking action and observation. But we need a new probability distribution (belief vector) for the next time step. Thus we multiply the updated distribution with transition matrix $P$.

The future expected (discounted) cost can be computed as follows:
\begin{align}
\mathbb{E}\{V_{t+1}(b_{t+1})|r_t,b_t\}&=\sum_{i=r_t}^M {b_t(i)} V_{t+1}(T_{r_t}[b_{t}]P)\nonumber\\
&+\sum_{i=0}^{r_t-1} {b_t(i)V_{t+1}(P_{i,.})}.\label{eq:vfnR3}
\end{align}
Note that for all $t=1,...,T$, $V_t(b_t)=\min_{\pi} J^{\pi}_{T-t}(b_t)$ with probability 1. In particular, $V_1(b_1)=J_T^{\pi}(b_1)$.
A policy $\pi^{opt}$ is optimal if for $t=1,...,T$; $r_t^{opt}(b_t)$ achieves the minimum in (\ref{eq:vfnmin}), denoted by: %%% =\pi^{opt}
\begin{align}%\label{eq:optimalv}
r_t^{opt}(b_t):=\arg\min_{r \in \mathcal{M}} V_t(b_t;r).\nonumber
\end{align}

This DP is equivalent to the DP given in Section \ref{sec:seqDP} as follows:
\begin{align}
&V_t(P_{i,.})=W_{t-1}(i), \nonumber\\
&V_0(b_0)=W(b_0),\nonumber\\
&r_{\tau}^{opt}(b_{\tau})=\underline{\pi}^{opt}_{s^{FO},t^{FO}}(\tau),\nonumber
\end{align}
where $b_{\tau}$ is the updated belief using (\ref{eq:pt1}) based on the optimal actions at times $t+1,...,\tau-1$ (propagated from $b_{t}=I_s$) and $t^{FO}$ is the last time of full observation to the state $s^{FO}$ before $\tau$.

\section{}\label{app:PBFO}

To prove Proposition \ref{prop:FO}, we use the belief-based DP presented in Appendix \ref{app_beliefDP} and it is enough to prove the following lemma:

\begin{lem}\label{prop:FO2}
The cost-to-go of the optimal policy is lower bounded by the cost-to-go of the full observation (FO) case under the same belief vector, \textit{i.e.},
\begin{align}\label{eq:FO2}
V_t(b_t) \geq V^{FO}_t(b_t).
\end{align}
\end{lem}

We need the following lemma about the concavity of cost-to-go functions to prove Lemma \ref{prop:FO2}.

\begin{lem}\label{L_convvf}
The expected discounted cost-to-go accrued under action $r$, $V_t(b;r)$, and the value function, $V_t(b)$, are concave with respect to the belief vector $b$,
\textit{i.e.}
\begin{align}\label{eq:convf}
&V_t(b;r)\geq \lambda V_t(b_1;r)+(1-\lambda)V_t(b_2;r),\ \ \forall r\in \mathcal{M}, \nonumber\\
&V_t(b)\geq \lambda V_t(b_1)+(1-\lambda)V_t(b_2),\ \ \forall 0\leq \lambda \leq 1.
\end{align}
\end{lem}

\begin{proof}[\textbf{Proof of Lemma \ref{L_convvf}}]
We use induction to prove the concavity of $V_t(b;r)$ with respect to the belief vector, $b$, for finite horizon.
Let's assume $b$ is a linear combination of two belief vectors $b_1$ and $b_2$, such that:
\begin{align}\nonumber%\label{eq:b}
b=\lambda b_1+(1-\lambda) b_2, 0\leq \lambda \leq 1.
\end{align}

At horizon $T$, the immediate cost, as given in (\ref{eq:vfnR}), is affine linear with respect to the belief vector. In other words,
\begin{align}\label{eq:affT}
\bar C(b;r)&=\lambda \bar C(b_1;r)+ (1-\lambda) \bar C(b_2;r).\end{align}
which confirms the concavity of the expected cost-to-go at horizon $T$.
Now assuming $V_{t+1}(.)$ is concave, we will consider $V_t(.)$. Using (\ref{eq:vfn}) and (\ref{eq:vfnR3}) we have:

\begin{align}
V_t(&b;r)-\lambda V_t(b_1;r) -(1-\lambda) V_t(b_2;r) \nonumber\\
&=[C(b;r)-\lambda \bar C(b_1;r)\nonumber\\
&- (1-\lambda) \bar C(b_2;r)] -(1-\lambda) V_{t+1}(T_r [b_2] P)\sum_{i=r}^{M} {b_2(i)})\nonumber\\ %\label{eq:conv12}\\
&+\beta(V_{t+1}(T_r [b] P)\sum_{i=r}^{M} {b(i)}-\lambda V_{t+1}(T_r [b_1] P)\sum_{i=r}^{M} {b_1(i)},\nonumber\\
\end{align}
and from (\ref{eq:affT}) and  $\lambda'=\lambda \frac{\sum_{i=r}^{M} {b_1(i)}}{\sum_{i=r}^{M} {b(i)}}$:
\begin{align}
V_t(b;r)&-\lambda V_t(b_1;r) -(1-\lambda) V_t(b_2;r) \nonumber\\
&=\beta \sum_{i=r}^{M} {b(i)} [V_{t+1}(T_r [b] P)-\lambda' V_{t+1}(T_r [b_1] P)\nonumber\\
&-(1-\lambda') V_{t+1}(T_r [b_2] P)], \nonumber% \label{eq:convl}
\end{align}

Let $j\geq r$:
\begin{align}%\label{eq:lamp}
&\lambda' T_r [b_1](j)+(1-\lambda') T_r [b_2](j)\nonumber\\
&=\frac{\lambda \sum_{i=r}^{M} {b_1(i)} T_r [b_1](j) + (1-\lambda)\sum_{i=r}^{M} {b_2(i)} T_r [b_2](j)}{\sum_{i=r}^{M} {b(i)}} \nonumber\\
&=\frac{1}{\sum_{i=r}^{M} {b(i)}}[\lambda \sum_{i=r}^{M} {b_1(i)} \frac{b_1(j)}{\sum_{i=r}^{M} {b_1(i)}} \nonumber\\
&+ (1-\lambda)\sum_{i=r}^{M} {b_2(i)} \frac{b_2(j)}{\sum_{i=r}^{M} {b_2(i)}}]
  \nonumber\\
&=\frac{\lambda b_1(j) + (1-\lambda)b_2(j)}{\sum_{i=r}^{M} {b(i)}}=\frac{b(j)}{\sum_{i=r}^{M} {b(i)}}=T_r [b](j),\nonumber
\end{align}
and for $j<r$, $T_r [b_1](j)+(1-\lambda') T_r [b_2](j)=0$.
Multiplying by $P$, we have $\lambda' T_r [b_1]P+(1-\lambda') T_r[b_2]P=T_r[b]P$. The induction step follows the concavity of $V_{t+1}(.)$.

To prove the concavity of value function, $V_t(b)$, with respect to $b$ we use the definition of (\ref{eq:vfnmin}) to get:
\begin{subequations}
\begin{align}
&V_t(b)=\min_{r}{V_t(b;r)}=V_t(b;r^*) \label{eq:pfl2}\\
&\geq \lambda V_t(b_1;r^*)+(1-\lambda)V_t(b_2;r^*)\label{eq:pfl21}\\
&\geq \lambda \min_{r_1}{V_t(b_1;r_1)}+(1-\lambda)\min_{r_2}{V_t(b_2;r_2)}\label{eq:pfl22}\\
&=\lambda V_t(b_1)+(1-\lambda)V_t(b_2),\label{eq:pfl23}
\end{align}
\end{subequations}
where $r^*=\arg \min_{r}{V_t(b;r)} $ and (\ref{eq:pfl21}) is the result of the lemma for $V_t(b;r^*)$ and applying the definition of (\ref{eq:vfnmin}) one more time in (\ref{eq:pfl23}) completes the proof.
\end{proof}
%\medskip

\begin{proof}[\textbf{Proof of Lemma \ref{prop:FO2}}]
First, the cost-to-go function of FO policy can be computed as:
\begin{align}
V^{FO}_t(b_t,r)=\bar C(b_t; r)+\beta \sum_{i=0}^{M} b_t(i) V^{FO}_{t+1}(P_{y,.}),\nonumber
\end{align}
Thus the optimal policy for FO policy is equivalent to the myopic policy since the second term is independent of the action $r$ and,
\begin{align}
\arg \min_r V^{FO}_t(b_t,r)=\arg \min_r \bar C(b_t; r).\nonumber
\end{align}
Therefore,
\begin{align}\nonumber
V^{FO}_t(b_t)=\bar C(b_t; \pi^{myopic}(b_t))+\beta \sum_{i=0}^{M} b_t(i) V^{FO}_{t+1}(P_{i,.}).
\end{align}
Now to prove the proposition we use induction. Fist, at $t=T$ we have:
\begin{align}
V_T(b_T)&- V^{FO}_T(b_T)\nonumber\\
&=\min_r \bar C(b_T; r)-\bar C(b_T; \pi^{myopic}(b_T))= 0,\nonumber
\end{align}
since the optimal action and the myopic action are equivalent at the horizon.
Now assuming (\ref{eq:FO}) is true at time steps $t+1$ onwards, we should prove it for time $t$.
\begin{align}
V_t(b_t)-& V^{FO}_t(b_t)=\bar C(b_t; r^{opt}(b_t))\nonumber\\
&+\beta \sum_{i=0}^{r^{opt}(b_t)-1}b_t(i) V_{t+1} (P_{i,.})\nonumber\\
&+\beta \sum_{i=r^{opt}(b_t)}^M b_t(i) V_{t+1} (T_{r^{opt}(b_t)}[b_{t}]P)\nonumber\\
&-\bar C(b_t; \pi^{myopic}(b_t))-\beta \sum_{i=0}^{M} b_t(i) V^{FO}_{t+1}(P_{i,.}), \nonumber
\end{align}
where $T_{r^{opt}(b_t)}[b_{t}]P=\frac{\sum_{j=r^{opt}(b_t)}^M b_t(j)P_{j,.}}{\sum_{i=r^{opt}(b_t)}^M b_t(i)}$.
We use the concavity of the value function to get the following inequality:
\begin{align}\label{eq:conc}
V_{t+1} (T_{r^{opt}(b_t)}[b_{t}]P) \geq
\frac{\sum_{i=r^{opt}(b_t)}^M b_t(i)V_{t+1} (P_{i,.})}{\sum_{j=r^{opt}(b_t)}^M b_t(j)}.
\end{align}
Therefore, by applying (\ref{eq:conc}) and merging two summations:
\begin{subequations}
\begin{align}
V_t(b_t)-& V^{FO}_t(b_t)\nonumber\\
&=[\bar C(b_T; r^{opt}(b_t))-\bar C(b_T; \pi^{myopic}(b_T))]\label{eq:FOind2-a}\\
&+\beta \sum_{i=0}^{M}b_t(i) [V_{t+1} (P_{i,.})-V^{FO}_{t+1}(P_{i,.})]\geq 0. \label{eq:FOind2-b}
\end{align}
\end{subequations}
The term in (\ref{eq:FOind2-a}) is greater than or equal to zero based on the definition of the myopic policy. The term in (\ref{eq:FOind2-b}) is also greater than or equal to zero using the induction assumption at $t+1$. Thus the whole expression is greater than or equal to zero and the proof is complete.
\end{proof}

\section*{Acknowledgement}
This work was supported in part by the U.S. National
Science foundation under ECCS-EARS awards numbered
1247995 and 1248017, by the Okawa foundation through
an award to support research on ``Network Protocols that
Learn'', and a partial support from L3-communications as well as UCSD's center for Wireless Communications and Networked Systems. Parisa Mansourifard holds AAUW American Dissertation Completion Fellowship for 2015-2016.

\bibliographystyle{IEEEtran}
\bibliography{myrefs}

\end{document}